\documentclass[pra,longbibliography,twocolumn,showpacs,nofootinbib,superscriptaddress,notitlepage,floatfix]{revtex4-2}
\usepackage{amsmath}
\usepackage{amssymb,bm}
\usepackage{amsthm}

\usepackage{qcircuit}

\newtheorem{innercustomgeneric}{\customgenericname}
\providecommand{\customgenericname}{}
\newcommand{\newcustomtheorem}[2]{
  \newenvironment{#1}[1]
  {
   \renewcommand\customgenericname{#2}
   \renewcommand\theinnercustomgeneric{##1}
   \innercustomgeneric
  }
  {\endinnercustomgeneric}
}

\newcustomtheorem{customthm}{Theorem}
\newcustomtheorem{definitionBob}{Definition}

\usepackage{color,dsfont} 
\usepackage{graphicx}
\usepackage{subcaption}
\usepackage{ragged2e}
\DeclareCaptionJustification{justified}{\justifying}
\usepackage[colorlinks=true, hyperindex, breaklinks, linkcolor=blue, urlcolor=blue, citecolor=blue]{hyperref}
\usepackage[normalem]{ulem}
\usepackage[capitalise]{cleveref}
\usepackage{mathrsfs}

\usepackage{mathtools}
\usepackage{float}
\usepackage{verbatim}
\usepackage{latexsym}
\usepackage{amsmath}
\usepackage{amssymb}
\usepackage{setspace}
\usepackage{amsfonts}
\usepackage{stmaryrd}
\usepackage{xcolor}
\usepackage{enumitem}
\usepackage{hhline}
\usepackage[normalem]{ulem}

\newtheorem{theorem}{Theorem} 
 
 \newtheorem{definition}{Definition}

 \newtheorem{protocol}{Protocol}

\newcommand{\ket}[1]{|#1\rangle}

\newcommand{\codepar}[1]{\ensuremath{[\![#1]\!]}}

\definecolor{azure}{rgb}{0.0, 0.5, 1.0}
\definecolor{blue-green}{rgb}{0.0, 0.67, 0.57}
\definecolor{neoncarrot}{rgb}{1.0, 0.64, 0.26}

\crefname{equation}{Eq.\!}{Eqs.\!}
\crefname{figure}{Fig.\!}{Figs.\!}
\usepackage{multirow}
\mathchardef\mhyphen="2D
\AtBeginDocument{\colorlet{defaultcolor}{.}}

\begin{document}

\title{Optimization tools for distance-preserving flag fault-tolerant error correction}

\author{Balint Pato}
\email{balint.pato@duke.edu}
\thanks{contributed equally}
\affiliation{
	Duke Quantum Center, Duke University, Durham, NC 27701, USA
}
\affiliation{
	Department of Electrical and Computer Engineering, Duke University, Durham, NC 27708, USA
}
\author{Theerapat Tansuwannont}
\email{t.tansuwannont.qiqb@osaka-u.ac.jp}
\thanks{contributed equally}
\thanks{present address: Center for Quantum Information and Quantum Biology, Osaka University, Toyonaka, Osaka 560-0043, Japan.}
\affiliation{
	Duke Quantum Center, Duke University, Durham, NC 27701, USA
}
\affiliation{
	Department of Electrical and Computer Engineering, Duke University, Durham, NC 27708, USA
}
\author{Shilin Huang}

\thanks{Present address: Department of Applied Physics, Yale University, New Haven, CT 06511, USA.}
\affiliation{
	Duke Quantum Center, Duke University, Durham, NC 27701, USA
}
\affiliation{
	Department of Electrical and Computer Engineering, Duke University, Durham, NC 27708, USA
}

\author{Kenneth R. Brown}
\email{ken.brown@duke.edu}
\affiliation{
	Duke Quantum Center, Duke University, Durham, NC 27701, USA
}
\affiliation{
	Department of Electrical and Computer Engineering, Duke University, Durham, NC 27708, USA
}
\affiliation{
	Department of Physics, Duke University, Durham, NC 27708, USA
}
\affiliation{
	Department of Chemistry, Duke University, Durham, NC 27708, USA
}

\begin{abstract}

	Lookup table decoding is fast and distance-preserving, making it attractive for near-term quantum computer architectures with small-distance quantum error-correcting codes. In this work, we develop several optimization tools that can potentially reduce the space and time overhead required for flag fault-tolerant quantum error correction (FTQEC) with lookup table decoding on Calderbank-Shor-Steane (CSS) codes. Our techniques include the compact lookup table construction, the Meet-in-the-Middle technique, the adaptive time decoding for flag FTQEC, the classical processing technique for flag information, and the separated $X$ and $Z$ counting technique. We evaluate the performance of our tools using numerical simulation of hexagonal color codes of distances 3, 5, 7, and 9 under circuit-level noise. Combining all tools can result in more than an order of magnitude increase in pseudothreshold for the hexagonal color code of distance 9, from $(1.34 \pm 0.01) \times 10^{-4}$ to $(1.42 \pm 0.12) \times 10^{-3}$.

\end{abstract}


\maketitle

\section{Introduction}

Inside a future large-scale quantum computer, there will be a continuous battle against unwanted interactions with the environment. The main goal of fault-tolerant quantum error correction (FTQEC) protocols \cite{Shor96} is to create a robust channel to transfer quantum information from the past to the future. The threshold theorem states that it is possible to suppress the failure rate of this channel (the logical error rate) to an arbitrarily small value given that the physical error rate of the constituent operations are below the accuracy threshold \cite{Knill98, K97,Preskill98,AB08,AGP06,AKP06}. It is essential to reduce both space and time overhead (the numbers of qubits and gates) for scalable quantum computing as decreasing logical error rates requires increasing overhead \cite{Steane03,PR12,CJL16b,TYC17}, and the current, leading proposals for FTQEC schemes have daunting requirements \cite{GE21,BMTSHKLSSV22}.

An FTQEC scheme is designed to be robust against propagating errors that emerge from faulty gates during the execution of the protocol \cite{Shor96}. The scheme also has to protect against ancilla preparation and measurement errors, usually through repeated syndrome measurements. For an \codepar{n,k,d} stabilizer code \cite{Gottesman97} which encodes $k$ logical qubits into $n$ physical qubits and has minimum distance $d$, Shor's solution \cite{Shor96} was to utilize cat-state ancilla register that requires $w$ ancilla qubits and $(d+1)^2/4$ rounds of syndrome measurements, where $w$ is the maximum weight of the stabilizer generators. In Steane-style syndrome extraction \cite{Steane96b}, the ancilla register requires $n$ qubits and is encoded with the same quantum error correcting code (QECC) as the data qubits. Similarly, in Knill-style error correction \cite{Knill05a}, the ancilla register consists of two blocks of $n$ qubits encoded in the same QECC as the data qubits.

In contrast to complex ancilla structures, bare ancillas can also be used to fault-tolerantly extract the syndrome while preserving the minimum distance for some specific families of stabilizer codes \cite{TS14, HB20, LGDHB17} and subsystem codes \cite{BDPS13,LMB18,HiggottB21}, or by tolerating some loss of distance \cite{LAR11,BKS21,ZWG23}. For a general stabilizer code, however, generator measurements with bare ancillas might not be possible.
A series of works aiming to reduce the size of the ancilla register resulted in increasingly lighter-weight constructions \cite{DA07,YK17}, which also led to the \textit{flag FTQEC} schemes for perfect codes of distance three that use only two ancillas per generator \cite{CR17a}, one flag-qubit and one syndrome qubit. The flag schemes later generalized to arbitrary codes of distance $d$ require $d+1$ ancillas per generator \cite{CR20}, while the schemes for some specific families of codes require fewer \cite{CB18,TCL20,CKYZ20,CZYHC20,TL21,TL22}.

FTQEC schemes based on extraction circuits with cat states and flag FTQEC schemes both require repetition of syndrome measurements that can result in a large number of gates. Adaptive syndrome measurement schemes in which the subsequent measurement procedures depend on the previous syndrome measurement outcomes have been explored to reduce the number of rounds required for FTQEC schemes with Shor-type extraction circuits \cite{Zalka96, DR20, TPB23}.

During the execution of an FTQEC protocol, faults can occur at any gate on any round of the syndrome measurements. The only information about the error on data qubits that we can obtain is a sequence of error syndromes, and we want to find an appropriate recovery operator from this information. An ideal strategy would be using all syndrome bits from all rounds, that is, the whole measurement outcomes in space-time. For some codes with a nice structure, such as surface codes, an efficient space-time decoder exists \cite{DKLP02}. However, constructing a space-time decoder for a general stabilizer code is not simple. To simplify the problem, we will consider an error decoder, which is composed of two parts: the space and the time decoders. Under the assumption that the syndrome measurements can be faulty, the time decoder finds a round of syndrome measurements that has no faults and gives a correct syndrome. The space decoder then uses the correct syndrome to construct a recovery operator.

Conventionally, flag FTQEC uses a lookup table decoder as a space decoder and relies on Shor-style repeated syndrome measurements as a time decoder, although there are instances where this is not the case \cite{CKYZ20_v3}. These decoders have pros and cons. The lookup table decoder is fast and distance-preserving. However, building a lookup table requires an exhaustive search over all possible fault combinations up to a certain number of faults, and the table requires a lot of memory to store. Thus, it might not work well with a code of high distance (unless code concatenation is applied). The Shor-style time decoder is simple and compatible with any space decoder. However, the large time overhead required in the repetition can result in a lower threshold.

In this work, we build several optimization tools for both space and time decoders for the purpose of reducing the overhead of both to obtain better-performing protocols for flag qubit-based FTQEC.
Most of our tools are applicable to general stabilizer codes, but we primarily focus on self-orthogonal CSS codes (CSS codes in which $X$- and $Z$-type generators are of the same form) in which the number of physical qubits is odd, the number of logical qubits is 1, and logical $X$ and $Z$ operators are transversal for simplicity. Our main results are the following: (1) We develop a technique to build a lookup table more efficiently. Our compact lookup table can leverage the structure of a self-orthogonal CSS code and requires 87.5\% less memory footprint compared to a lookup table designed for a generic stabilizer code. Our method also efficiently verifies whether a configuration of the flag circuits preserves the code distance. The development also leads to a notion of \textit{fault code}, which can be useful in error sampling for the circuit-level noise model. (2) We introduce the Meet-in-the-Middle (MIM) technique, which can help the lookup table decoder correct faults more than the number of errors correctable by the underlying code. Although the correction is not always successful, the higher success probability can significantly increase the pseudothreshold in our simulations. (3) We generalize previous work \cite{TPB23} on adaptive syndrome measurement schemes to flag FTQEC and introduce one-tailed and two-tailed adaptive time decoders, which are useful in different circumstances. We also develop a classical processing technique on flag information that makes our FTQEC protocols compatible with any fault-tolerant Clifford computation. (4) We use our optimization tools and perform numerical simulations on the hexagonal color codes \cite{BM06} of distances 3, 5, 7, and 9. The results show that each of our tools can significantly reduce the logical error rates and increase the pseudothreshold for each code while preserving the code distance. For the hexagonal color code of distance 9, the pseudothreshold is improved by one order of magnitude, from $(1.34\pm 0.01) \times 10^{-4}$ to $(1.42 \pm 0.12) \times 10^{-3}$, when all techniques are applied.

This paper is organized as follows. In \cref{sec:background}, we define the noise model in this work, review flag FTQEC, and provide definitions of fault-tolerant error correction. In \cref{sec:space_decoder}, we develop optimization tools for space decoder, including an efficient method to build a compact lookup table and the MIM technique. In \cref{sec:time_decoder}, we develop optimization tools for time decoder, including the one-tailed and two-tailed adaptive time decoder, and other extended techniques for CSS codes. In \cref{sec:numerical results}, we provide numerical results for the hexagonal color codes and observe the effects of the MIM, the adaptive time decoding, and the separated $X$ and $Z$ counting techniques on the logical error rates. We discuss and conclude our results in \cref{sec:discussions}.

\section{Background}  \label{sec:background}

Quantum systems are fragile and can lose their properties easily when interacting with the environment. To protect quantum information, one can use a QECC to encode the quantum data. Quantum error correction (QEC) is a process that identifies an error when it occurs  and then applies an appropriate error correction (EC) operator to remove the error. However, quantum operations in the process can be faulty and may introduce more errors to the system. For this reason, we want to make sure that the QEC process is \emph{fault tolerant}, which provides robustness guarantees against the impact of noise on the error correction implementations.

In this section, we first describe the noise model that will be used in this work and provide the conventional definition of fault-tolerant error correction in \cref{subsec:FTQEC_def}. We then review flag FTQEC and provide a revised definition of fault-tolerant error correction, which is more suitable for flag FTQEC in \cref{subsec:flag_FTQEC}.

\subsection{Noise model and conventional definition of fault-tolerant error correction}
\label{subsec:FTQEC_def}

An \codepar{n,k,d} stabilizer code \cite{Gottesman97} encodes $k$ logical qubits using $n$ physical qubits and can correct up to $\tau=\lfloor (d-1)/2 \rfloor$ errors, where $d$ is the code distance. A stabilizer code is described by a \emph{stabilizer group}, an Abelian group generated by $r=n-k$ commuting Pauli operators whose elements are called stabilizers. The code space is the simultaneous +1 eigenspace of all elements in the stabilizer group.

The QEC process for a stabilizer code can be done by first measuring the eigenvalues of all stabilizer generators. An $r$-bit string of measurement outcomes is called the error syndrome (where bits 0 and 1 refer to $+1$ and $-1$ eigenvalues of each generator). An example of a circuit for measuring an eigenvalue of a stabilizer generator is displayed in \cref{fig:bare}. After the syndrome is obtained, an appropriate recovery operator will be found by a mapping called \emph{error decoder}. Finally, the recovery operator will be applied to the data qubits. For Calderbank-Shor-Steane (CSS) codes \cite{CS96,Steane96b}, it is possible to correct $X$- and $Z$-type errors separately. In this work, we follow \textit{standard CSS decoding} \cite{DT14}, meaning independent recovery for $X$- and $Z$-type errors, thus not taking the effect of $X/Z$ correlations like $Y$ errors into account.

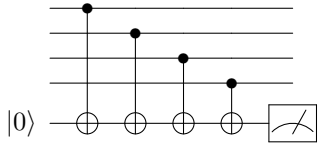
\begin{figure}[tbp]
	\begin{equation}
		\Qcircuit @C=1em @R=.7em {
		& \ctrl{4} & \qw & \qw & \qw & \qw & \\
		& \qw & \ctrl{3} & \qw & \qw & \qw & \\
		& \qw & \qw & \ctrl{2} & \qw & \qw & \\
		& \qw  & \qw & \qw & \ctrl{1} & \qw & \\
		\lstick{\ket{0}} & \targ & \targ & \targ & \targ & \meter
		} \nonumber
	\end{equation}
	\caption{A syndrome extraction circuit with bare ancilla for measuring a stabilizer generator of the form $ZZZZ$.}
	\label{fig:bare}
\end{figure}

If all gates in the syndrome measurement process (with an example circuit in \cref{fig:bare})  are perfect, a stabilizer code of distance $d$ should be able to correct up to $\tau$ errors as desired. However, the process above may not be fault-tolerant under the circuit-level depolarizing noise. This is because a single faulty gate may lead to an error that can propagate to multiple errors on the data qubits, often referred to as hook errors \cite{DKLP02}. These errors can always be handled by complex ancilla \cite{Shor96,Steane97,HB21} or flag circuits \cite{CR17a} and sometimes handled by the circuit order \cite{TS14, HB20, LGDHB17,LGDHB17}.

In this work, we use the \emph{circuit-level depolarizing noise model}.  After each gate, a \textit{fault} occurs on the support of the gate. Every single-qubit gate is followed by a single-qubit Pauli operator $P \in \{X,Y,Z\}$ with probability $p/3$ each, and every two-qubit gate is followed by a two-qubit Pauli operator $P_1\otimes P_2 \in \{I,X,Y,Z\}^{\otimes 2} \setminus \{I\otimes I\}$ with probability $p/15$ each. In addition, a single-qubit preparation and measurement can also be faulty; this is modeled be a bit-flip channel after a single-qubit preparation or before a single-qubit measurement with error probability $p$.

One way to define FTQEC is by using the definition proposed by Aliferis, Gottesman, and Preskill:

\begin{definition}{Fault-tolerant error correction \cite{AGP06}}

	Let $t \leq \lfloor (d-1)/2\rfloor$ where $d$ is the distance of a stabilizer code. An error correction protocol is \emph{$t$-fault tolerant} if the following two conditions are satisfied:
	\begin{enumerate}
		\item Error correction correctness property (ECCP): For any input codeword with error of weight $r$, if $s$ faults occur during the protocol with $r+s \leq t$, ideally decoding the output state gives the same codeword as ideally decoding the input state.
		\item Error correction recovery property (ECRP): If $s$ faults occur during the protocol with $s \leq t$, regardless of the weight of the error on the input state, the output state differs from any valid codeword by an error of weight at most $s$.
	\end{enumerate}
	\label{def:FT_AGP}
\end{definition}

When a QEC protocol satisfies \cref{def:FT_AGP}, it is guaranteed that the output error will have weight $\leq t$ whenever the weight of the input error plus the total number of faults in the protocol is $\leq t$. This means that if the next round of QEC has no faults, it can always correct the output error from the current round. Normally, we would like to construct an FTQEC protocol in which $t$ is as close to $\tau=\lfloor (d-1)/2\rfloor$. If $t=\tau$, we say that \emph{the FTQEC protocol preserves the code distance}.

\subsection{Flag technique and revised definition of fault-tolerant error correction}
\label{subsec:flag_FTQEC}

Before describing the flag technique for FTQEC, let us consider a well-known Shor FTQEC \cite{Shor96} applied to a stabilizer code of distance $d$. In this scheme, a stabilizer generator of weight $w$ is measured using a cat state of the form $\frac{1}{\sqrt{2}}(\ket{0}^{\otimes w}+\ket{1}^{\otimes w})$ and transversal CNOT gates; see \cref{fig:Shor}. A circuit of this kind will be called a \emph{Shor syndrome extraction circuit}. When the cat state is prepared fault-tolerantly, a single fault in the circuit can lead to an error of weight no more than one on the data qubits, so the set of all possible errors arising from up to $t$ faults is exactly the same as a set of all possible errors on $\leq t$ qubits in this case. Therefore, any syndrome can uniquely identify the error (up to a multiplication of some stabilizer) when the number of faults in the protocol is  $\leq t$.

\begin{figure}[tbp]
	\begin{equation}
		\ \ \ \ \ \ \ \ \ \ \ \Qcircuit @C=1em @R=.7em {
		& \qw & \ctrl{4} & \qw & \qw & \qw & \qw & \qw  \\
		& \qw & \qw & \ctrl{4} & \qw & \qw & \qw & \qw  \\
		& \qw & \qw & \qw & \ctrl{4} & \qw & \qw & \qw  \\
		& \qw & \qw  & \qw & \qw & \ctrl{4} & \qw & \qw  \\
		& \gate{H} & \targ & \qw & \qw & \qw & \qw & \meter \\
		& \gate{H} & \qw & \targ & \qw & \qw & \qw & \meter \\
		& \gate{H} & \qw & \qw & \targ & \qw & \qw & \meter \\
		& \gate{H} & \qw  & \qw & \qw & \targ & \qw & \meter \inputgroupv{5}{8}{1em}{3.1em}{\mkern-65mu\frac{\ket{0000}+\ket{1111}}{\sqrt{2}}}\\
		} \nonumber
	\end{equation}
	\caption{A Shor syndrome extraction circuit for measuring a stabilizer generator of the form $ZZZZ$.}
	\label{fig:Shor}
\end{figure}
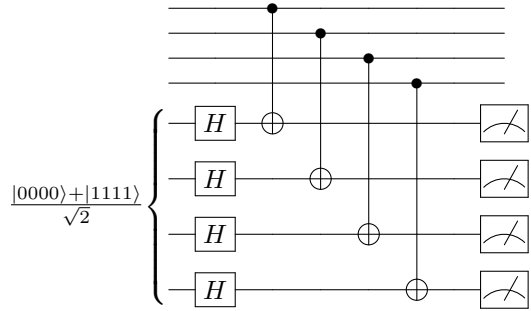

One drawback of the Shor syndrome extraction circuit is that the number of required ancilla qubits is equal to the maximum weight of the stabilizer generators. Also, fault-tolerant preparation of the ancilla cat state requires verification \cite{Shor96} or Divincenzo-Aliferis ancilla decoding circuit \cite{DA07}, which requires additional space and time overhead.
One possible technique that can reduce the number of required ancillas for FTQEC is the flag technique \cite{CR17a}, in which each syndrome extraction circuit uses one ancilla qubit to keep the syndrome measurement outcome and a few flag ancillas to find a location that a fault might have occurred.
A circuit of this kind will be called a \emph{flag circuit};
See \cref{fig:flag} for an example.
The flag measurement outcomes give extra information that can be used to partition set of all possible errors from a certain number of faults. Therefore, it is possible to distinguish between two non-equivalent errors that correspond to the same syndrome if the flag measurement outcomes associated with each error are different, making error correction easier.

Here we define fault combination, fault set, and distinguishability of a fault set as follows:

\begin{definition}{Fault combination, combined data error, and cumulative flag vector \cite{TL22}}

	A \emph{fault combination} $\Lambda =\{\lambda_{1},\lambda_{2},\dots,\lambda_{r}\}$ is a set of $r$ faults $\lambda_{1}, \lambda_{2}, \cdots, \lambda_{r}$.
	Suppose that the Pauli error due to the fault $\lambda_{i}$ can propagate through the circuit and lead to \emph{data error} $E(\lambda_{i})$ and \emph{flag vector} $\vec{f}(\lambda_{i})$. The \emph{combined data error} $\mathbf{E}(\Lambda)$ and \emph{cumulative flag vector} $\vec{F}(\Lambda)$ corresponding to the fault combination $\Lambda$ are,
	\begin{align}
		\mathbf{E}(\Lambda) & =\prod_{i=1}^r E(\lambda_{i}), \label{eq:combined_E}                           \\
		\vec{F}(\Lambda)    & =\sum_{i=1}^r \vec{f}(\lambda_{i})\;(\mathrm{mod}\;2). \label{eq:cumulative_f}
	\end{align}
	\label{def:fault_combi}
\end{definition}

\begin{definition}{Distinguishable fault set \cite{TL22}}

	Let $\mathcal{S}$ be the stabilizer group of a stabilizer code, and let the \emph{fault set} $\mathcal{F}_t$ denote the set of all possible fault combinations arising from up to $t$ faults during the measurement of stabilizer generators of $\mathcal{S}$.
	We say that $\mathcal{F}_t$ is \emph{distinguishable} if for any pair of fault combinations $\Lambda_p,\Lambda_q \in \mathcal{F}_t$, at least one of the following conditions is satisfied:
	\begin{enumerate}
		\item $\vec{s}(\mathbf{E}(\Lambda_p)) \neq \vec{s}(\mathbf{E}(\Lambda_q))$, or
		\item $ \vec{F}(\Lambda_p) \neq \vec{F}(\Lambda_q)$, or
		\item $\mathbf{E}(\Lambda_p) = \mathbf{E}(\Lambda_q)\cdot M$ for some stabilizer $M \in \mathcal{S}$,
	\end{enumerate}
	where $\vec{s}(\mathbf{E})$ is the error syndrome of a combined error $\mathbf{E}$.
	Otherwise, we say that $\mathcal{F}_t$ is \emph{indistinguishable}.
	\label{def:distinguishable}
\end{definition}

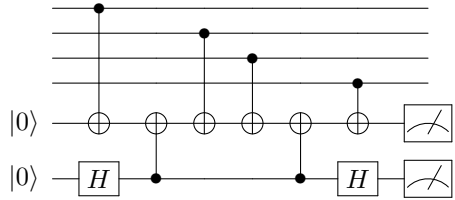
\begin{figure}[tbp]
	\begin{equation}
		\Qcircuit @C=1em @R=.7em {
		& \ctrl{4} & \qw & \qw & \qw & \qw & \qw & \qw & \\
		& \qw & \qw & \ctrl{3} & \qw & \qw & \qw & \qw & \\
		& \qw & \qw & \qw & \ctrl{2} &  \qw & \qw & \qw & \\
		& \qw & \qw & \qw & \qw & \qw & \ctrl{1} & \qw & \\
		\lstick{\ket{0}} & \targ & \targ & \targ & \targ & \targ & \targ & \meter \\
		\lstick{\ket{0}} & \gate{H} & \ctrl{-1} & \qw & \qw & \ctrl{-1} & \gate{H} & \meter
		} \nonumber
	\end{equation}
	\caption{A flag circuit for measuring a stabilizer generator of the form $ZZZZ$. }
	\label{fig:flag}
\end{figure}

Note that the cases of faulty flag qubit measurements are included when the fault set is calculated for verifying fault set distinguishability (see \cref{subsection:compact lookup}).
Having a distinguishable fault set is a key to successful error decoding.
Given a set of syndrome extraction circuits (with or without flags), we can calculate the fault set $\mathcal{F}_t$ and check whether it is distinguishable. If it is, all possible errors arising from up to $t$ faults that correspond to the same syndrome and cumulative flag vector are always logically equivalent.
Therefore, if the syndrome measurements give a syndrome $\vec{s}$ and a cumulative flag vector $ \vec{F}$, we can pick any error that corresponds to the pair $ (\vec{s},\vec{F})$ to be a recovery operator. Using this idea, a decoding table and an FTQEC protocol can be constructed.

With the notion of fault distinguishability, it is possible to further generalize the definition of FTQEC as follows:

\begin{definition}{Fault-tolerant error correction (revised) \cite{TL22}}

	Let $t \leq \lfloor (d-1)/2\rfloor$ where $d$ is the distance of a stabilizer code. An error correction protocol is \emph{$t$-fault tolerant} if the following two conditions are satisfied:
	\begin{enumerate}
		\item ECCP: For any input codeword with an error that can arise from $r$ faults before the protocol and corresponds to the zero cumulative flag vector, if $s$ faults occur during the protocol with $r+s \leq t$, ideally decoding the output state gives the same codeword as ideally decoding the input state.
		\item ECRP: If $s$ faults occur during the protocol with $s \leq t$, regardless of the number of faults that can cause the input error, the output state differs from any valid codeword by an error that can arise from $s$ faults and corresponds to the zero cumulative flag vector.
	\end{enumerate}
	\label{def:FT_TL}
\end{definition}

The main difference between these two definitions of FTQEC is that \cref{def:FT_TL} considers the \emph{number of faults} that can cause the input (or the output) error instead of the \emph{weight} of the error. An FTQEC protocol satisfying \cref{def:FT_TL} can be constructed if we can find syndrome extraction circuits that give a distinguishable fault set (see previous results \cite{TL22} by one of the authors of this work and the discussion in the next section for more details). In fact, while the threshold theorem proved by Aliferis, Gottesman and Preskill~\cite{AGP06} relied on the weight of the error to define fault tolerance (\cref{def:FT_AGP}), the theorem has been shown to hold \cite{TL22} even if the definition of fault tolerance uses the number of faults (\cref{def:FT_TL}) instead. For flag FTQEC, using \cref{def:FT_TL} can result in simpler FTQEC protocols, so we will use \cref{def:FT_TL} in the protocol development throughout this work.

\section{Optimization Tools for Space decoding}
\label{sec:space_decoder}

In this work, the term space decoder refers to a process that finds a recovery operator from a given syndrome under the assumption that it is exactly the same as the syndrome of an error that occurred to the codeword. The decoder succeeds if multiplying the error and the recovery operator gives a trivial logical operator (a stabilizer), and it fails if the multiplication gives a nontrivial logical operator. Our goal is to develop a space decoder such that whenever the total number of faults in the whole protocol is $\leq t$, the decoder always succeeds.
In this work, we are interested in a lookup table-based space decoder for flag FTQEC, so the decoder will use both syndrome and flag information obtained during the syndrome measurements. Note that the ability to correct faults for a certain code depends on the structure of the circuits for syndrome extraction, such as the ordering of gates.

In this section, we develop optimization tools for space decoding. In \cref{subsection:compact lookup}, we discuss how to efficiently construct a lookup table for error decoding for a distinguishable fault set $\mathcal{F}_t$, and introduce the notion of fault code. In \cref{subsec:MIM}, we discuss the Meet-in-the-Middle technique, an additional technique that can help improving our space decoders for both codes and increase the accuracy of the decoding.

\subsection{Compact lookup table for minimum weight decoding and fault code}\label{subsection:compact lookup}

In this section, we discuss how to construct the fault set $\mathcal{F}_t$, verify its distinguishability, and construct the lookup table for error decoding. With our method, we can reduce the memory footprint requirement of the lookup table by 87.5\% for self-orthogonal CSS codes compared to a lookup table designed for generic stabilizer codes. We also present the framework of \emph{fault codes} that enables fast construction using streamlined Pauli-frame simulation represented as matrix algebra operations over $\mathrm{GF}(2)$.

A brief summary of our methods is as follows. Let the \emph{weight of a fault combination} be the number of faults that give rise to the fault combination. The decoding table maps each full syndrome $ (\vec{s}, \vec{F})$ to a recovery operator that corresponds to the combined data error of the minimum-weight fault combination that results in the full syndrome. To construct the decoding table, we start by collecting all weight-1 fault combinations that may arise in the extraction circuits. We map each resulting full syndrome to its corresponding data error. At this point, we say that the \textit{search radius} of the lookup table is 1. Afterward, we combine pairs of weight-1 fault combinations to create all possible weight-2 fault combinations. The combined data error of each weight-2 fault combination is obtained by simply taking the product of the data errors, and the full syndrome is obtained by adding full syndromes of the weight-1 fault combination modulo 2. If the combining process leads to a new syndrome, we store it in the table. If the process leads to an existing syndrome, we have a collision and do one of the following: (1) If the stored combined data error and the new combined data error are the same up to a stabilizer, then we do nothing. (2) If the stored combined data error and the new combined data error differ by a logical operator (up to a stabilizer), then we raise an error; this implies that $\mathcal{F}_2$ is not distinguishable. At this point, if there is no combination that causes the second case (that is, $\mathcal{F}_2$ is distinguishable), we say that the search radius of the lookup table is 2. We can gradually increase the search radius using similar ideas until we reach the maximum search radius in which the fault set is distinguishable. Here, we rely on an efficient representation of the combined data errors using a decomposition of Pauli operators to pure errors, stabilizers and logical operators \cite{Poulin_2006}.

During sampling, the decoder receives a full syndrome that was measured. When the decoder finds this syndrome in the lookup table, it returns the corresponding \emph{actual recovery operator} (ARO). However, when the decoder cannot find the syndrome in the lookup table, it only returns a so-called \textit{canonical recovery operator} (CRO). Each syndrome has a unique canonical recovery operator, which guarantees that applying such an operator to the erroneous encoded state will map it back to the code space but with a possible logical error.

The full description of our methods is presented below.

\subsubsection{Reducing the memory footprint}

To decode an \codepar{n,k,d} stabilizer code, we can construct a lookup table that, for all possible fault combinations of weight 0 to $t$ (where $t=\lfloor \frac{d-1}{2} \rfloor$), stores the full syndrome $ \vec{\sigma}=(\vec{s}, \vec{F})$ as the key and maps the combined data error as the recovery operator. While this approach works, it is expensive. Let $T_{stab}$ denote the number of distinct full syndromes for the fault combinations of weight 0 to $t$ for a generic stabilizer code. As $T_{stab}$ and thus the size of the lookup table grow exponentially in $n$, $n-k$ (the number of generators), and the number of circuit locations, we want to choose a representation to store data as efficiently as possible. For general stabilizer codes, $n-k$ bits are required for the syndrome bits and $n-k$ bits for the cumulative flag vector (assuming flag circuits with single flag ancilla for simplicity). Meanwhile, the recovery operator requires $2n$ bits using the symplectic representation. Thus we have $T_{stab}(4n-2k)$ bits of data in the map.

Leveraging the structure of CSS codes, we can significantly improve the memory footprint. Assuming standard CSS decoding in which two separate lookup tables are used for $X$ and $Z$ decoding. Denote with $r_X$ and $r_Z$ the number of $X$- and $Z$-type stabilizer generators satisfying $r_X+r_Z=n-k$. The per entry cost decreases, as the entries only need to cater for $X$- or $Z$-type operators and syndromes. Each entry for the $X$- and $Z$-type syndromes will have $2r_X$ and $2r_Z$ bits respectively for the syndrome and the cumulative flag vector, and $n$ bits for the recovery operator. A self-orthogonal CSS code needs only one table futher decreasing the cost; see more detail on the total number of bits in \cref{app:memory footprint}. Moreover, we can reduce the number of bits for the recovery operator to $k$ using the following two key ideas:
\begin{enumerate}
	\item In general, for an \codepar{n,k,d} code, each Pauli operator $P \in \mathcal{P}_n$ can be decomposed as a product $P = EML$ of a pure error $E$, a stabilizer $M \in \mathcal{S}$, and a logical operator $L \in \overline{\mathcal{P}}_k$ (where $\overline{\mathcal{P}}_k$ is the $k$-dimensional logical Pauli group) \cite{Poulin_2006}.  We define a fixed set of pure errors called \emph{canonical recovery operators} (CRO), one CRO for each unique syndrome $\vec{s}$.
	\item Given a syndrome $\vec{s}(E)$, the goal of decoding is to find a recovery operator $R$ such that $RE \in \mathcal{S}$, thus $R$ converts the error into the logical identity operation.  For any possible Pauli error, we only have to store its \textit{logical class}, a value that indicates how the error is related to a CRO with the same syndrome. This enables the map value to be only $2k$ bits of information in general and $k$ bits in case the code is a self-orthogonal CSS code. In this latter case and with $k=1$, the logical class is 0 if the multiplication of the Pauli operator and the CRO with the same syndrome is in the stabilizer group, otherwise the logical class is 1.

\end{enumerate}
Altogether, for a self-orthogonal CSS code with $n \gg k$, the size of the table can be as small as 12.5\% of the table if we viewed the code as a generic stabilizer code and stored the full recovery operators instead of the logical classes. For a CSS code that is not self-orthogonal, the gain is smaller but still significant. See \cref{app:memory footprint} for detailed calculations around savings for the lookup table. Note that if the lookup table is used for proving distinguishability, all unique syndromes are required. However, in a realtime decoding architecture, the entries corresponding to significantly low probability fault combinations may be excluded, resulting in further reduction \cite{DLJ22}.

\subsubsection{Constructing the lookup table}

We now explicitly describe an algorithm to construct the lookup table. During the construction of the lookup table, we have a systemic way to enumerate fault combinations with their full syndromes and combined data errors instead of running through a circuit simulator for each case. The exhaustive enumeration of all possible fault combinations of weight 0 to $t$ is done in two steps. First, we enumerate the single faults and capture their full syndrome and logical class in a single column of the \emph{fault check matrix}, $H_{\text{f}}$ using matrix algebra over $\mathrm{GF}(2)$ to represent the propagation of errors in our syndrome extraction circuits. Second, we combine these columns in all possible combinations of $0$ to $t$ faults ($\sum_{i=0}^t {N \choose i}$ combinations in total, where $N$ is the number of possible single faults) while keeping track of the weight of each fault combination. This last step verifies whether $\mathcal{F}_t$ is distinguishable (which is equivalent to verifying whether the protocol is distance preserving), and at the same time builds a lookup table for the decoder.

\emph{Enumerating weight-1 faults---\;}From here on, we will only consider a self-orthogonal CSS code, and denote its parity check matrix $H$.
In order to list all possible single faults under the circuit-level depolarizing noise model, it is sufficient to consider all possible weight-1 faults within a single round of syndrome measurements.
Each column of the fault check matrix $H_{\text{f}}$ describes for each possible weight-1 fault what its full syndrome and its logical class are.
As the logical class of each fault depends on how its CRO is defined, we define the fault check matrix relative to the right inverse $H^{-1}$ of $H$ (for which $HH^{-1}=I_{(n-k)/2}$).

The high-level structure of $H_{\text{f}}$ consists of three major groups of rows and three major groups of columns. The three groups of rows are the $(n-k)/2$ \textit{generator bits}, the $(n-k)/2$ \textit{flag bits}, and the $k$ bits for the logical class. Each single fault which is represented by a column of $H_{\text{f}}$ can be put into one of the following three categories:

\begin{enumerate}

	\item \textbf{Pure data qubit errors} that result only in generator bits. They do not trigger flags, resulting in all-zero flag bits. The CRO $R$ of each pure data qubit error $E$ can be described by each column of $H^{-1}H$ (since the syndromes of CROs are $H(H^{-1}H) = (HH^{-1})H = H$),  thus the product $RE$ of each $E$ can be described by each column of $I_{n}\oplus H^{-1}H$ (where the matrix addition, denoted by $\oplus$, and multiplication are over $\mathrm{GF}(2)$). If $E$ is an $X$-type (or a $Z$-type) error, the logical class of $RE$ is described by a $k$-bit string in which the $i$-th bit indicates whether $RE$ anticommutes with $\bar{Z}_i$ (or $\bar{X}_i$). That is, the logical classes of all pure data qubit errors are described by
	      \[
		      \left(\begin{array}{c}
				      J_1^T(I_{n}\oplus H^{-1}H) \\
				      J_2^T(I_{n}\oplus H^{-1}H) \\
				      \dots                      \\
				      J_k^T(I_{n}\oplus H^{-1}H)
			      \end{array}
		      \right),
	      \]
	      where $J_i$ is the column vector representing $\bar{Z}_i$ (or $\bar{X}_i$).
	\item \textbf{Flag ancilla preparation or measurement errors} which do not propagate to data qubits, thus, each single-flag error will result in a single flag bit.
	      Therefore, all errors of this type have the all-zero syndrome and logical class 0, while the flag bits can be easily represented by the $(n-k)/2 \times (n-k)/2$ identity matrix.
	\item \textbf{Gate faults} that cause errors on the syndrome ancilla which can propagate to data and flag qubits --- we order these faults by top-down and left-right place of occurrence and capture their effect in syndrome bits, flag bits, and logical class.
	      The part of the effective matrix corresponding to this type of faults is denoted by $H_\text{f,gate}$.
\end{enumerate}

Note that single measurement and reset errors on the syndrome ancilla are ignored during this analysis as their effects would be removed by the time decoder through the repetition of syndrome measurements.

Generalizing $H_{\text{f}}$ for a non-self-orthogonal CSS code is straightforward. In that case, the parity check matrices for $X$- and $Z$-type errors can be different, leading to different fault check matrices. Generalizing $H_{\text{f}}$ for a generic stabilizer code is more complicated but still doable, as all operators must be considered in the symplectic form. In that case, the number of rows for the logical class is $2k$. Also, instead of taking the inner product with $J_i$, whether each CRO commutes or anticommutes with each logical operator can be determined by the symplectic inner product between the symplectic bitstrings representing the CRO and the logical operator.


In the case that the code is a self-orthogonal CSS code, $n$ is odd, $k=1$, and logical $X$ and logical $Z$ operators are transversal, the fault check matrix is,
\[
	H_{\text{f}} = \left( \begin{array}{c|c|c}
			H                       & 0           & \multirow{3}{*}{$H_{\text{f,gate}}$} \\
			0                       & I_{(n-1)/2} &                                      \\
			J^T(I_n \oplus H^{-1}H) & 0           &
		\end{array} \right),
\]
where $J$ is the all-one column vector of length $n$ (representing $X^{\otimes n}$ or $Z^{\otimes n}$).

As an example, consider the first group of columns for the \codepar{7,1,3} Steane code \cite{Steane96b} whose stabilizer generators can be defined by the parity check matrix,
\[
	H = \begin{pmatrix}
		0 & 0 & 0 & 1 & 1 & 1 & 1 \\
		0 & 1 & 1 & 0 & 0 & 1 & 1 \\
		1 & 0 & 1 & 0 & 1 & 0 & 1 \\
	\end{pmatrix}.
\]
One can pick its right inverse $H^{-1}$ as follows:
\[
	H^{-1} = \begin{pmatrix}
		0 & 0 & 1 \\
		0 & 1 & 0 \\
		0 & 0 & 0 \\
		1 & 0 & 0 \\
		0 & 0 & 0 \\
		0 & 0 & 0 \\
		0 & 0 & 0
	\end{pmatrix}.
\]
We can see that each column of $H^{-1}$ gives a Pauli operator for each syndrome bit. For a data error $E$ of any weight, the syndrome $\vec{s}(E)=HE$ can be recovered with the CRO defined by $R(\vec{s}(E))\equiv H^{-1}\vec{s}(E)$ as
\begin{align}
	\vec{s}(E \oplus R) & =H(E\oplus R) \nonumber        \\
	                    & =HE \oplus HH^{-1}HE \nonumber \\
	                    & =HE \oplus HE = 0. \nonumber
\end{align}
For example, for $E=(0110000)^T$, $\vec{s}(E)=(001)^T$ and $R(\vec{s}(E))=(1000000)^T$, thus $RE=(1110000)^T$, for which the syndrome is trivial (as $RE$ is a logical operator).

For errors of weight 1 on the data qubits, the operator $RE$ of each error can be represented by each column of $I_n \oplus H^{-1}H$. Since the logical class of $RE$ can be determined by its weight parity, the logical classes of this type of errors are the row of $L \equiv J^T(I_n \oplus H^{-1} H)$ where $J$ is the all-one column vector. That is for the Steane code, the part of $H_{\text{f}}$ corresponding to pure data qubit errors is
	\[
		\begin{pmatrix}
			H \\
			\hline
			0 \\
			\hline
			J^T(I_n \oplus H^{-1} H)
		\end{pmatrix}
		=
		\begin{pmatrix}
			0 & 0 & 0 & 1 & 1 & 1 & 1 \\
			0 & 1 & 1 & 0 & 0 & 1 & 1 \\
			1 & 0 & 1 & 0 & 1 & 0 & 1 \\
			\hline
			0 & 0 & 0 & 0 & 0 & 0 & 0 \\
			0 & 0 & 0 & 0 & 0 & 0 & 0 \\
			0 & 0 & 0 & 0 & 0 & 0 & 0 \\
			\hline
			0 & 0 & 1 & 0 & 1 & 1 & 0
		\end{pmatrix}.
	\]

	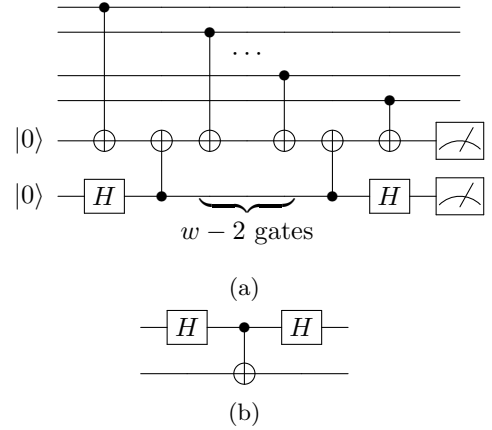
\begin{figure}[tbp]
		\begin{subfigure}[b]{0.45\textwidth}
			\begin{equation}
				\ \ \ \ \ \ \Qcircuit @C=1em @R=.7em {
				& \ctrl{5} & \qw       & \qw      & \qw             & \qw       & \qw       & \qw       & \qw    \\
				& \qw      & \qw       & \ctrl{4} & \qw             & \qw       & \qw       & \qw       & \qw    \\
				&          &           &          & \cdots          &           &           &           & \\
				& \qw      & \qw       & \qw      & \qw             & \ctrl{2}  & \qw       & \qw       & \qw    \\
				& \qw      & \qw       & \qw      & \qw             & \qw       & \qw       & \ctrl{1}  & \qw    \\
				\lstick{\ket{0}} & \targ    & \targ     & \targ    & \qw             & \targ     & \targ     & \targ     & \meter \\
				\lstick{\ket{0}} & \gate{H} & \ctrl{-1} & \qw      & \qw              & \qw       & \ctrl{-1} & \gate{H}  & \meter
				\gategroup{1}{4}{7}{6}{.7em}{_\}} \\
				&          &           &          & \text{$w-2$ gates}        &           &           &           &
				} \nonumber
			\end{equation}
			\captionsetup{justification=centering}
			\caption{}
			\label{subfig:flag_w}
		\end{subfigure}
		\begin{subfigure}[b]{0.45\textwidth}
			\begin{tabular}{c}
				\Qcircuit @C=1em @R=.7em {
				 & \gate{H} & \ctrl{1} & \gate{H} & \qw \\
				 & \qw      & \targ    & \qw      & \qw
				} \nonumber
			\end{tabular}
			\captionsetup{justification=centering}
			\caption{}
			\label{subfig:XNOT}
		\end{subfigure}
		\caption{(a) A flag circuit for measuring a $Z$-type stabilizer generator of weight $w$ in this work. A flag circuit for measuring a $X$-type stabilizer generator of weight $w$ can be obtained by replacing each CNOT gate that connects the data qubit to the syndrome ancilla with the gate in (b).}
		\label{fig:flag_w}
	\end{figure}

	\emph{Constructing $H_\text{f,gate}$---\;}In this work, we focus on the case that any $Z$-type or $X$-type stabilizer generator of weight $w$ is measured using a flag circuit with a single flag ancilla similar to the circuit in \cref{fig:flag_w} (with a slight modification, similar construction for a general flag circuit can also be made). For $H_\text{f,gate}$, we are interested in how the errors propagate from the syndrome ancilla to the data qubits and the flag ancilla. The error propagation is represented via a binary matrices, an idea closely related to the ``gate matrix'', where the direction of propagation is the opposite way, towards the ancilla from the data qubits \cite{HB21,HB21b}. Given single-flag syndrome extraction circuits for all stabilizer generators and the CNOT ordering for each circuit, $H_\text{f,gate}$ can be calculated via the \textit{propagator matrix} $P$ and the \emph{aggregator matrix} $A$, defined as follows: For the error correction protocol with $n$ data qubits and $r$ flag bits (which is the same number as the number of $X$ or $Z$ stabilizer generators), The matrix $P$ has $n+r$ rows. The number of columns of $P$ is $\sum_{i=1}^r (w(g_i) + 2)$, where $w(g_i)$ is the Hamming weight of the $i$-th stabilizer generator ($g_i$). This is from the fact that for each CNOT gate in the single-flag syndrome extraction circuits, the only fault that can lead to a unique data error after propagation is the fault that leads to a single $Z$ error on the target qubit of the CNOT (which is the syndrome ancilla). To simplify the construction, we construct a submatrix $P_i$ of size $(n+r) \times (w(g_i)+2)$ for each row $g_i$ of $H$ (i.e each stabilizer generator), then concatenate the submatrices to get,
	\begin{align}
		P = \left( P_1 P_2 \ldots P_r \right).
	\end{align}
	As the order of the CNOT gates matters in subtle ways, for a given stabilizer generator $g_i$, we represent the CNOT ordering by the permutation $\pi_i: \{1,2,\dots,w(g_i)\}\rightarrow supp(g_i)$, where $\pi_i(j)$ indicates the control (data) qubit of the $j$-th CNOT (the target qubit is always the syndrome ancilla). $\pi_i$ can also be represented by a list. For example, two possible permutations of CNOT gates in the syndrome extraction circuit for measuring $g_1$ of the \codepar{7,1,3} Steane code are $\pi_1= [4,5,6,7]$ and $\pi_1= [4,6,5,7]$.

	To construct $P_i$, we iterate from $j=1$ to $w(g_i)$, and create a column for each iteration with all zeros except for the 1 in row $\pi_i(j)$. We then insert an all-zero column on the second from the left and the second from the right positions (which represent the flag CNOTs), and set its value to 1 at row $n+i$. In our running example of $g_1=(0001111)$, for a permutation of $\pi_1=[4,6,5,7]$,

	\begin{align}
		P_i = \begin{pmatrix}
			      0 & 0 & 0 & 0 & 0 & 0 \\
			      0 & 0 & 0 & 0 & 0 & 0 \\
			      0 & 0 & 0 & 0 & 0 & 0 \\
			      1 & 0 & 0 & 0 & 0 & 0 \\
			      0 & 0 & 0 & 1 & 0 & 0 \\
			      0 & 0 & 1 & 0 & 0 & 0 \\
			      0 & 0 & 0 & 0 & 0 & 1 \\
			      0 & 1 & 0 & 0 & 1 & 0 \\
			      0 & 0 & 0 & 0 & 0 & 0 \\
			      0 & 0 & 0 & 0 & 0 & 0
		      \end{pmatrix}.
	\end{align}
	The aggregator matrix $A$ plays the role of propagating the errors to the end of the syndrome measurement circuits. For each $g_i$, we define $A_i$ to be a square matrix of size $(w(g_i)+2) \times (w(g_i)+2)$ having a lower triangle set to all 1s, and define $A=\bigoplus_{i=1}^r A_i$ to be the direct sum of all $A_i$'s. In our example case of $g_1$,
	\begin{align}
		A_i = \begin{pmatrix}
			      1 & 0 & 0 & 0 & 0 & 0 \\
			      1 & 1 & 0 & 0 & 0 & 0 \\
			      1 & 1 & 1 & 0 & 0 & 0 \\
			      1 & 1 & 1 & 1 & 0 & 0 \\
			      1 & 1 & 1 & 1 & 1 & 0 \\
			      1 & 1 & 1 & 1 & 1 & 1 \\
		      \end{pmatrix}
	\end{align}
	Multiplying the propagator and the aggregator matrices yields,
	\begin{align}
		PA = \begin{pmatrix} \Omega \\ \hline \Phi \end{pmatrix}=\begin{pmatrix} \Omega_1 & \Omega_2 & \ldots & \Omega_r \\
                \hline
                \Phi_1   & \Phi_2   & \ldots & \Phi_r\end{pmatrix},
	\end{align}
	where columns of the submatrices $\Omega_i$ are the final Pauli operators, and columns of the submatrices $\Phi_i$ are the cumulative flag vectors after measuring $g_i$ and having a fault propagated from the syndrome ancilla to the data qubits at the location corresponding to the given column.

	Next, we find the syndromes for these Pauli operators by multiplying them with the parity check matrix,
	\begin{align}
		S = H \Omega.
	\end{align}
	Then, for each syndrome, we define the CRO based on the right inverse $H^{-1}$,
	\begin{align}
		\Theta = H^{-1} S.
	\end{align}
	Finally, we determine the logical class $L$ for each of the faults by adding the parity of the CRO and the propagated, final data error,
	\begin{align}
		L = J^T (\Omega \oplus \Theta).
	\end{align}
	As a result, the part of $H_{\text{f}}$ corresponding to the gate faults is,
	\begin{align}
		H_\text{f,gate} = \begin{pmatrix} S \\ \hline \Phi \\ \hline L \end{pmatrix} = \begin{pmatrix} H \Omega \\ \hline \Phi \\ \hline J^T(I_n \oplus H^{-1}H)\Omega \end{pmatrix}.
	\end{align}

	The relationship between the full syndrome, the data error, and the CRO of each fault is as follows: Suppose that the $i$-th column of $H_{\text{f}}$ (which represents a single fault on the $i$-th location) contains error syndrome $\vec{s}_i$, flag vector $\vec{f}_i$, and logical class $l_i$. The CRO of the fault is $H^{-1}\vec{s}_i$, while the data error of the fault is $l_iJ\oplus H^{-1}\vec{s}_i$. That is, in case of a single fault, the \emph{actual recovery operator} (ARO) we need to apply when finding the full syndrome $(\vec{s}_i,\vec{f}_i)$ is $l_iJ\oplus H^{-1}\vec{s}_i$.

		\emph{Verifying distinguishability and building the lookup table---\;}The fault set $\mathcal{F}_t$ is distinguishable if and only if there is no fault combination from up to $2t$ faults that gives a non-trivial logical operator with trivial full syndrome \cite{TL22}. As the fault check matrix already contains all the possible single faults, in case of $t=1$, we only need to extend the matrix by a column with all zeros (which represents 0 faults) and check whether there is a pair of columns which are the same except for the logical class. If there is, the combined data errors of one or two faults add up to an undetectable logical operator, meaning that $\mathcal{F}_1$ is not distinguishable.

		\begin{table*}[tbp]
			\begin{tabular}{|l|l|l|l|l|}
				\hline
				                             & \codepar{7,1,3}    & \codepar{19,1,5}    & \codepar{37,1,7}             & \codepar{61,1,9}        \\
				\hline
				\# columns of $H_{\text{f}}$ & 28                 & 88                  & 181                          & 307                     \\
				\# unique columns            & 20                 & 62                  & 128                          & 218                     \\
				\# fault combinations        & 20                 & 1,953               & 349,632                      & 93,263,997              \\
				Cache size                   & 20                 & 1,587               & 262,500                      & 67,166,572              \\
				Memory                       & $\leq 1\text{kB}$  & $\leq 1\text{kB}  $ & $\approx 50 \text{MB}     $  & $\approx 1.38\text{GB}$ \\
				Verification time            & $\leq 1\text{ms} $ & $\leq 1\text{ms}  $ & $ \approx 720 \text{ms}    $ & $\approx 58.9\text{s}$  \\
				\hline
			\end{tabular}
			\caption{Metrics of the lookup table. The number of columns of the fault check matrix counted in the first row results from the three-part structure of data errors, flag errors, and gate faults. These columns are not necessarily unique, which can be seen in the second row that counts the number of unique columns. The time to verify distinguishability for the different codes on a single thread with our C++ code depends on the number of unique columns, hence the verification of the higher distance code takes longer than shorter ones. All timings are reported using Intel Xeon Gold 6226R, 2.90GHz processors. Some fault combinations have the same full syndrome, hence the cache size is smaller than the full number of fault combinations. The cache size in memory is reported from actual usage, including the overhead of the hash table implementation.}
			\label{tab: decoder numbers}

		\end{table*}

		When $t\geq 2$, we populate the cache with the logical classes of higher-weight fault combinations by simply combining all possible fault combinations of lower-weight fault combinations while keeping track of the weights of the fault combinations. We describe the $i$-th fault combination as a key-value pair $[(\vec{s}_i,\vec{F_i}): (l_i, w_i)]$ where $(\vec{s}_i,\vec{F_i})$ is the full syndrome, $l_i$ is the logical class, and $w_i$ is the weight of the fault combination. Combining the $i$-th and the $j$-th fault combinations gives $[(\vec{s}_i\oplus \vec{s}_j,\vec{F_i}\oplus \vec{F_j}): (l_i\oplus l_j, w_i+w_j)]$. As we aim to check whether $\mathcal{F}_t$ is distinguishable, we fill up the cache by combining any pair of fault combinations that satisfy $w_i+w_j \leq t$. In case that the process gives the new key (the full syndrome) that already exists in the cache, we have a key conflict. This can be one of the following cases:
		\begin{enumerate}
			\item The new and the existing fault combinations have the same full syndrome and the same logical class but have different weights. In this case, we store the fault combination with smaller weight in the cache.
			\item  The new and the existing fault combinations have the same full syndrome but have different logical classes. As the sum of weights of these two fault combinations is $\leq 2t$, we raise an error---there exists a fault combination from up to $2t$ faults that gives a non-trivial logical operator with trivial full syndrome, that is, $\mathcal{F}_t$ is not distinguishable.
		\end{enumerate}

		If at the end we find that $\mathcal{F}_t$ is distinguishable, we can construct a lookup table of search radius $t$ from the cache as follows: for each key-value pair $[(\vec{s}_i,\vec{F_i}): (l_i, w_i)]$ in the cache, we store a new key-value pair $[(\vec{s}_i,\vec{F_i}): l_iJ\oplus H^{-1}\vec{s}_i]$ in the lookup table (the weights are not necessary for decoding, though it might be useful for estimating the number of faults that causes the full syndrome). That is, $l_iJ\oplus H^{-1}\vec{s}_i$ is the ARO for the full syndrome $(\vec{s}_i,\vec{F_i})$. When performing error decoding, the ARO is applied if the full syndrome obtained from measurements is found on the lookup table; otherwise, the CRO ($H^{-1}\vec{s}_i$) is applied.

		The lookup table can then be stored in an efficient binary format on disk or memory as needed. In \cref{tab: decoder numbers}, we displayed the metrics related to the lookup table decoder obtained by the algorithm above.

		In summary, we perform an exhaustive search of fault combinations which gives us a lookup table with search radius $t$; this is equivalent to verifying the distinguishability of $\mathcal{F}_t$. If we can construct the lookup table with $t=\tau=\lfloor (d-1)/2 \rfloor$, we have a minimum-weight decoder that is distance-preserving under the circuit-level depolarizing noise model. As a hash-table requires $\mathcal{O}(1)$ amortized complexity for lookup, this decoder is also relatively fast for numerical simulations or real-time decoding compared to more complicated algorithms such as MaxSAT-decoding \cite{Berent_Burgholzer_Derks_Eisert_Wille_2023}, neural-network-based decoder \cite{Maskara_Kubica_Jochym-OConnor_2019}, or the restriction decoder \cite{Delfosse_2014} with minimum weight perfect matching decoding \cite{higgott2023sparse}, all of which have at least $\mathcal{O}(n)$ complexity. However, the table size scales exponentially in the number of qubits, locations, and stabilizer generators, and thus, constructing the lookup table may be impractical for a code of high distance.

		\subsubsection{The fault code}

Any CSS code can be defined by its parity check matrix $H$, which maps a bitstring representing a combination of \emph{errors} on the data qubits to the error syndrome of the error combination. In the case of flag FTQEC where the circuit-level noise model is considered, we can use similar ideas and define a \textit{fault code} by the fault check matrix $H_{\text{f}}$ which maps a bitstring representing a combination of \textit{possible faults} to the full syndrome of the fault combination (which includes the error syndrome of the combined data error and cumulative flag vector) and the logical class relative to the CRO for the syndrome. It should be noted that the distance of the fault code might be lower than the distance of the underlying CSS code; this depends on the syndrome extraction circuits, which affect the distinguishability of the fault set. We can define the \emph{effective distance} $d_\mathrm{eff}$ to be the minimum number of faults that can give a fault combination with a non-trivial logical operator and the trivial full syndrome. The number of faults $t_\mathrm{eff}$ that the fault code can correct is $t_\mathrm{eff}=\lfloor (d_\mathrm{eff}-1)/2 \rfloor$ (this is the maximum number of $t$ in which $\mathcal{F}_t$ is distinguishable). If the effective distance and the code distance are equal, we say that the error correction protocol is distance preserving. Calculating the distance of classical codes can be done by determining the spark of the parity check matrix $H$, which is known to be NP-hard, in general~\cite{TP13}. However, the spark algorithm does not work in the case of degenerate CSS codes, as it reports only the minimum weight of the stabilizers which is a lower bound on the code distance \cite{CS96}. Our algorithm described in this section can be viewed as a modified spark algorithm that uses the logical class information to calculate the distance of the code (based on $H$) and also the effective distance of the fault code (based on $H_\textrm{f}$).

The perspective of the fault code can also be useful to extend a technique frequently used for error sampling (in qecsim \cite{qecsim_2021}, for example) to the circuit-level noise model beyond the code capacity noise model (memory errors only) and phenomenological noise model (both memory and measurement errors). Here, a randomly generated column vector of Hamming weight $w$ now represents faults on $w$ locations instead of errors on $w$ qubits. Suppose that the vector $\vec{v}$ represents the fault combination and $H_{\text{f}}\vec{v}$  gives the full syndrome $(\vec{s}(\vec{v}),\vec{F}(\vec{v}))$ and the logical class $l(\vec{v})$. In an error correction simulation, the decoder can predict the recovery operator $\vec{r}$ based on the full syndrome.
We will find that the predicted recovery operator causes a logical error if and only if $l(\vec{v})$ and $l(\vec{r})$ differ.

In principle, this method can lead to a better sampling rate compared to running the full circuit simulation for each sample. However, one needs to be aware of the probability distribution when generating vectors representing the fault combinations, as each possible single fault might not occur at the same rate.

		\subsection{Meet-in-the-Middle Technique}
		\label{subsec:MIM}

		If the fault set $\mathcal{F}_t$ of each code is distinguishable, the flag FTQEC protocol can correct up to $t$ faults with certainty. However, whenever $t+1$ or more faults occur, the error correction is not guaranteed; our decoder can either remove the error or cause a logical error on the encoded state. Although the probability of having $t+1$ or more faults is $\mathcal{O}(p^{t+1})$, being able to correct more cases of faults can lead to a higher pseudothreshold. In this section, we introduce the Meet-in-the-Middle technique, which can help correct errors in case there are more than $t$ faults in our FTQEC protocol. Note that this technique is general and could help any FTQEC protocol with a table-based decoder to correct faults more than its capability if the stabilizer code being used is not a perfect (or a perfect CSS) code.

		\begin{figure*}[tbp]
			\centering
			\begin{subfigure}[b]{0.4\textwidth}
				\includegraphics[width=\textwidth]{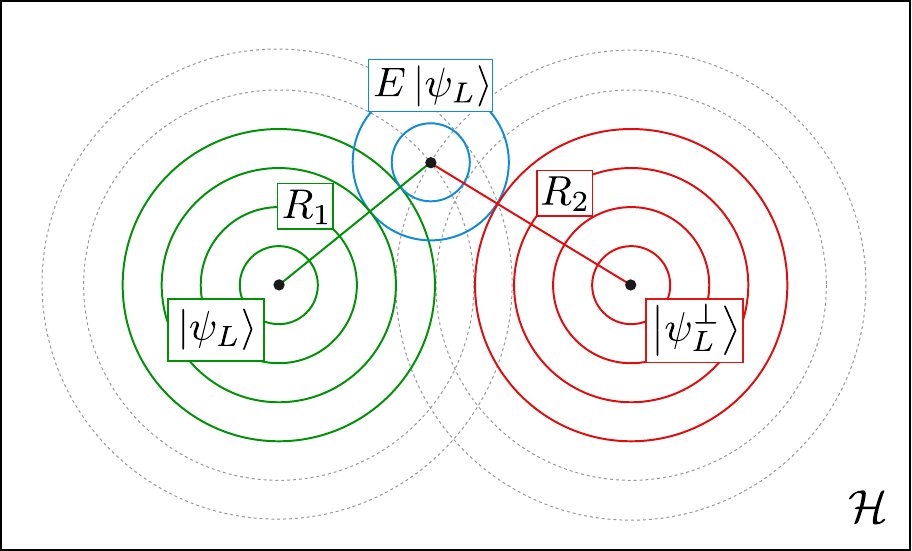}
				\captionsetup{justification=centering}
				\caption{}
				\label{subfig:MIM_good}
			\end{subfigure}
			\hspace{4mm}
			\begin{subfigure}[b]{0.4\textwidth}
				\includegraphics[width=\textwidth]{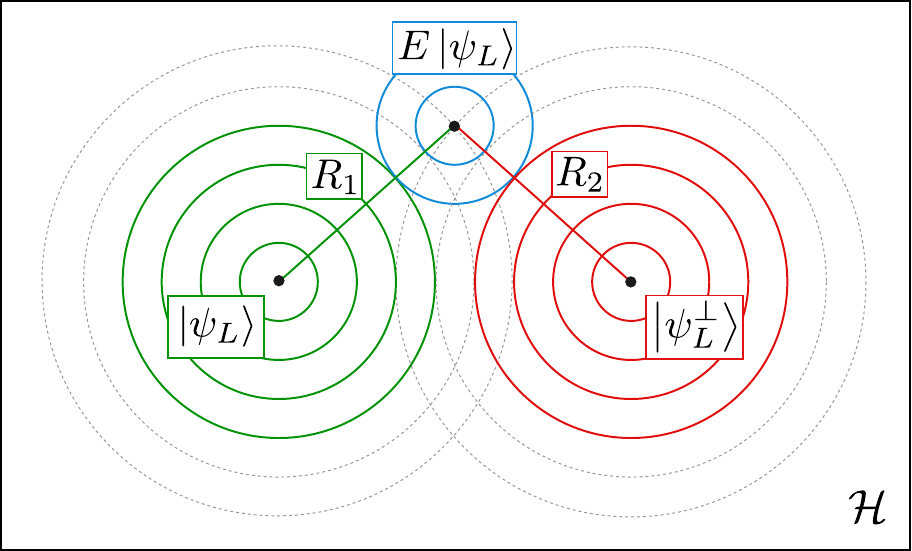}
				\captionsetup{justification=centering}
				\caption{}
				\label{subfig:MIM_bad}
			\end{subfigure}
			\caption{An illustration of the error decoding using a lookup table and the MIM technique on the Hilbert space $\mathcal{H}=\mathbb{C}^{\otimes n}$ of the physical qubits. A code of distance 9 is considered in this example. Using a lookup table with search radius 4 only, any erroneous states lying on the green (or red) circles, which are up to $4$ faults (circles) away from the logical state $\ket{\psi_L}$ (or $\ket{\psi_L^\perp}$), will be recovered to the logical state $\ket{\psi_L}$ (or $\ket{\psi_L^\perp}$). Consider the erroneous state $E\ket{\psi_L}$ which is not on any green or any red circle. In (a), $E\ket{\psi_L}$ is 5 faults away from $\ket{\psi_L}$ and 6 faults away from $\ket{\psi_L^\perp}$. Using the MIM table of radius 1, the recovery operator found by the decoder is $R_1$. Since $R_1E$ is a stabilizer, $R_1$ brings the state back to the original state $\ket{\psi_L}$. In (b), $E\ket{\psi_L}$ is 6 faults away from both $\ket{\psi_L}$ and $\ket{\psi_L^\perp}$. Using the MIM table of radius 2, the recovery operator found by the decoder is either $R_1$ such that $R_1E$ is a stabilizer, or $R_2$ such that $R_2E$ is a nontrivial logical operator. In this case, the state after recovery can be either $\ket{\psi_L}$ or $\ket{\psi_L^\perp}$.}
			\label{fig:MIM}
		\end{figure*}

		The Meet-in-the-middle (MIM) technique is inspired by the bidirectional search algorithm \cite{Korf_2003} to improve the table-based decoder previously discussed in \cref{subsec:flag_FTQEC} (see also \cref{subsection:compact lookup}) in case that the decoder cannot find in its lookup table the full syndrome obtained from measurements. Consider the case that the fault set $\mathcal{F}_t$ is distinguishable, and a lookup table of search radius $t$ can be constructed. Suppose that more than $t$ faults occur and the full syndrome is $(\vec{s}_m,\vec{F}_m)$ which is not in the lookup table. The table-based decoder discussed in \cref{subsec:flag_FTQEC} will return the canonical recovery operator that may cause a logical error after correction. To make successful error correction more probable in such cases, one could, in principle, construct a lookup table with a search radius larger than $t$ by relaxing the distinguishability requirement for fault combinations with weights higher than $t$. However, this can be impractical as the number of fault combinations grows too fast when the search radius increases.

		To overcome this issue, we instead conduct a \emph{search during decoding} starting from the missing syndrome $(\vec{s}_m,\vec{F}_m)$. That is, we construct another decoding table, called the MIM table, with search radius at most $\rho \leq t$ using ideas similar to the original lookup table, but we also add $(\vec{s}_m,\vec{F}_m)$ to the map key before storing the syndrome in the MIM table and check whether it is in the decoding lookup table or not. If a new map key in the MIM table is the same as some map key in the decoding lookup table, the search stops, and the decoder constructs a recovery operator from two combined data errors from the MIM table and the decoding table that correspond to the two map keys. If the MIM search radius reaches $\rho$ and no matching syndrome is found, the decoder returns the CRO for the full syndrome. Using the recovery operator obtained from this method, we can correct up to $t+\rho$ faults with probability higher than using the CRO of the full syndrome only.

		An example of error decoding using a lookup table and the MIM technique is illustrated in \cref{fig:MIM}. In our FTQEC protocols for hexagonal color codes of distance 3, 5, 7, and 9, we find that constructing the MIM table with search radius $\rho=t$ is sufficiently fast to be used at runtime. Note that the MIM technique does not guarantee successful error correction due to potential degeneracy in syndromes above the guaranteed number of correctable faults. However, we do find numerically that the MIM technique has a positive impact on the performance of our decoders for the hexagonal color codes for distances 3, 5, 7, and 9.

		\section{Optimization Tools for Time decoding} \label{sec:time_decoder}

		In general, faults can happen at any point during syndrome measurements, and the syndrome obtained at each round of measurements may not be the \emph{correct syndrome} (the syndrome of the combined data error at the end of that round). In particular, measurement errors can lead to a syndrome that differs from the correct syndrome by some bits. Errors on the data or ancilla qubits that happen in the middle of the syndrome extraction can also result in a syndrome that only captures some parts of the correct syndrome. Applying a space decoder to a faulty syndrome can lead to an incorrect recovery operation. For this reason, one must perform multiple rounds of syndrome measurements.

		The goal of a time decoder is to find a round with a correct syndrome at least at one point in the whole syndrome measurement process. If this can be done, an FTQEC protocol satisfying both conditions in \cref{def:FT_TL} can be constructed. Note that according to the definition, it is sufficient to consider only the case that the total number of faults in the whole protocol is no more than $t$, where $t$ is the number of errors that a stabilizer code being used can correct. This is because the failure probability of the FTQEC protocol (the probability of having $t+1$ or more faults in the protocol) will be $\mathcal{O}(p^{t+1})$ similar to the failure probability of an ideal error correction with the same stabilizer code. (Nevertheless, in terms of better decoding accuracy, it is beneficial to consider correcting some cases of $t+1$ or more faults as suggested by the MIM technique in \cref{subsec:MIM}.)

		In this section, we develop several types of time decoders for flag FTQEC, building on the ideas of adaptive decoders for Shor-style error correction \cite{TPB23}. Different time decoders use different fault count estimation procedures. In \cref{subsec:Shor_time_decoder}, we describe a conventional way to perform repeated syndrome measurements for flag FTQEC in terms of difference vectors, which will be useful for the development in latter sections. In \cref{subsec:adaptive_time_decoder}, we develop one-tailed and two-tailed adaptive time decoders which utilize flag information in the protocols. One-tailed adaptive decoder is applicable to a larger family of codes, while the two-tailed adaptive decoder is more optimized to self-orthogonal CSS codes but need to be used with an extended technique so that it becomes fully fault tolerant when applying to quantum computation. In \cref{subsec:ext_CSS}, we develop two extended techniques that can further improve the performance of our adaptive time decoders for FTQEC, given that the code being used is a self-orthogonal CSS code.

		\subsection{Shor time decoder for flag FTQEC} \label{subsec:Shor_time_decoder}

		In Shor's original approach \cite{Shor96}, the syndrome extraction is repeated until the same syndrome appears $t+1$ times in a row. Observe that for $R$ repeated but untrustworthy syndromes, at least $R$ faults are required to make them the same (we can think of having exactly the same measurement errors for example).  Therefore, to make sure that a round with a correct syndrome exists when considering the case with up to $t$ faults, it is sufficient to wait for $t+1$ repeated measurements. A time decoder with this stop condition will be referred to as \emph{Shor time decoder}.

		It is possible to rephrase the Shor time decoder using the notion of difference vector. For a syndrome history $(\vec{s}_1,\vec{s}_2,\dots,\vec{s}_m)$ of length $m$, we define a \emph{difference vector} $\vec{\delta}$ to be an $(m-1)$-bit string in which $\delta_i$ is 0 if $\vec{s}_{i+1}=\vec{s}_i$, or $\delta_i$ is 1 if $\vec{s}_{i+1}\neq\vec{s}_i$. As two repeated syndrome measurements are represented by a zero in the difference vector, Shor's method can be reformulated as waiting for $t$ consecutive zeros in $\vec{\delta}$.

		As we aim to correct no more than $t$ faults, the analysis of our time decoders can be made easier by thinking about the budget of $t$ faults. Shor's method spends all of this budget on counting consecutive zeros in the difference vector and is completely oblivious to other parts of the syndrome history (because the counter is reset whenever bit one appears). We call the parts of the syndrome history outside of the zero substring the \textit{context} of the zero substring. As Shor's method does not take the context into account, we call this strategy ``context-unaware''. In the worst-case scenario for the Shor time decoder, $(t+1)^2$ rounds of syndrome measurements are done before the stopping condition is satisfied. The context of the zero substrings contains useful information and not counting the faults in the context results in underestimating the number of faults that can cause a given syndrome history. Context-aware strategies that have a better estimate of the number of faults can stop earlier and execute fewer measurements, resulting in higher pseudothresholds.

		As flag circuits are used in the syndrome extraction, we also obtain a flag vector history $(\vec{f}_1,\vec{f}_2,\dots,\vec{f}_m)$ from $m$ rounds of syndrome measurements, which also leads to a cumulative flag vector history $ (\vec{F}_1,\vec{F}_2,\dots,\vec{F}_m)$.
		Note that the calculation of a difference vector \emph{does not} involve flag vectors; since the cases of faulty flag qubit measurements are considered when we evaluate the distinguishability of a fault set, all flag measurement outcomes are considered correct and can be used for error decoding. Our goal is to find a round such that all syndrome bits are correct.

		The correct syndrome will be used in conjunction with the flag information obtained \emph{right before} the measurements of the correct syndrome. Suppose that the code being used is a CSS code, and $X$-type generator measurements at round $i$ (which lead to $\vec{s}_{i,x}$, $\vec{f}_{i,x}$, and $ \vec{F}_{i,x}$) are done before $Z$-type generator measurements (which lead to $\vec{s}_{i,z}$, $\vec{f}_{i,z}$, and $ \vec{F}_{i,z}$). If the syndrome from round $l$ is correct according to Shor time decoder, $Z$-type (or $X$-type) error correction will be done using $\vec{s}_{l,x}$ and $ \vec{F}_{l-1,z}$ (or $\vec{s}_{l,z}$ and $ \vec{F}_{l,x}$). We also use similar ideas for error correction with other time decoders.

		Suppose that a table-based space decoder for flag FTQEC can be constructed (as discussed in \cref{sec:space_decoder}). Then, a flag FTQEC protocol with Shor time decoder is as follows:

		\begin{protocol}{Flag FTQEC protocol with Shor time decoder}

			Let $t=\lfloor (d-1)/2\rfloor$ be the number of errors that a stabilizer code of distance $d$ can correct. Let $\vec{s}_i=(\vec{s}_{i,x},\vec{s}_{i,z})$ and $\vec{f}_i=(\vec{f}_{i,x},\vec{f}_{i,z})$ be syndrome and flag vector obtained from the $i$-th round of full syndrome measurements with flag circuits. Let the cumulative flag vector at the $i$-th round be $ \vec{F}_i=(\vec{F}_{i,x},\vec{F}_{i,z})=\sum_{j=1}^i \vec{f}_{j}\;(\mathrm{mod}\;2)$. After the $i$-th round with $i\geq 2$, calculate $\delta_{i-1}$. Repeat syndrome measurements until the last $t$ bits of $\vec{\delta}$ is zero or the total number of rounds reaches $(t+1)^2$. Suppose that the latest round is round $l$. Perform $Z$-type error correction using $ (\vec{s}_{l,x},\vec{F}_{l-1,z})$, and perform $X$-type error correction using $ (\vec{s}_{l,z},\vec{F}_{l,x})$.
			\label{pro:flag_shor}
		\end{protocol}

		\subsection{Adaptive time decoder for flag FTQEC}
		\label{subsec:adaptive_time_decoder}

		Recently, FTQEC protocols with adaptive syndrome measurement techniques have been proposed by some of the authors of this work \cite{TPB23}.
		Instead of using flag qubits, in that work, each stabilizer generator is measured using a syndrome extraction circuit with a cat state (similar to Shor's original circuits \cite{Shor96}). The authors show that using the \textit{adaptive strong decoder}, it is possible to reduce the number of syndrome measurement rounds in the worst-case scenario from $(t+1)^2$ rounds to $(t+3)^2/4-1$ rounds. The resulting FTQEC protocol satisfies the error-weight based definition of FTQEC (\cref{def:FT_AGP}) and is applicable to any stabilizer code. In this work, we extend the adaptive strong decoder based measurement techniques to flag FTQEC and develop  protocols satisfying the revised FTQEC conditions that use the number of faults instead of the weight of errors (\cref{def:FT_TL}). The main difference from Ref. \cite{TPB23} is that this work also uses flag information to estimate the number faults occurred in the protocol, leading to a faster procedure to find a syndrome suitable for error correction. We start by describing the key ideas of Ref. \cite{TPB23} in terms of correlated and uncorrelated bit histories, which is useful for bounding from below the number of occurred faults. Afterwards, we explain how each technique in Ref. \cite{TPB23} could be improved using the flag information.

		\subsubsection{\label{ssec: context counting} Counting faults in correlated and uncorrelated bit histories}

		Let us first consider a way to estimate the number of occurred faults from a given difference vector $\vec{\delta}$. A single fault can cause either one or two consecutive bits of ones in $\vec{\delta}$ \cite{TPB23}. Thus, for each substring $\vec{\kappa}$ in $\vec{\delta}$, the number of faults that can cause such a substring is bounded from below by the number of 11 sequences plus the number of remaining 1s in $\vec{\kappa}$.

		Suppose that the difference vector is of the form $\vec{\delta}=\eta_1 1 \eta_2 1 \dots 1 \eta_c$ where $\eta_j=00\dots00$ are zero substrings and $1 \leq j \leq c$. For each $\eta_j$ of length $\gamma_j \geq 1$ with $2 \leq j \leq c-1$, we define $\alpha$ to be the total number of non-overlapping $11$ sequences plus the total number of remaining 1s before the substring $1\eta_j 1$, and define $\beta_j$ similarly but for the substring after $1\eta_j 1$ (for $\eta_1$ and $\eta_c$, $\beta_1$ and $\alpha_c$ are defined similarly to those of other $\eta_j$'s, and we let $\alpha_1=0$, $\beta_c=0$). The zero substring $\eta_j$ of length $\gamma_j$ corresponds to $\gamma_j+1$ consecutive rounds with the same syndrome, so the number of rounds that can cause these rounds to give incorrect syndromes is at least $\gamma_j+1$. Therefore, under the assumption that there are at most $t$ faults in the whole protocol, if we find that there exists $\gamma_j$ such that $t-\alpha_j-\beta_j < \gamma_j+1$, the syndromes of the $\gamma_j+1$ rounds that give rise to $\eta_j$ cannot be all incorrect. For this reason, at least one syndrome corresponding to $\eta_j$ is correct and can be used for error correction (see the full analysis in Ref. \cite{TPB23} for more details).

		For example, assume that the total number of faults in the protocol is $t=4$, and ten rounds of syndrome measurements give the following $\vec{\delta}$:
		\begin{align*}
			\text{Round:}        & \ 1 \ 2 \ 3 \ 4 \ 5 \ 6 \ 7 \ 8 \ 9 \ 10 \\
			\vec{\delta}\text{:} & \ \ 1 \ 1 \ 0 | 1 \ 0 \ 0 \ 1 | 0 \ 1
		\end{align*}
		Focusing on the substring $1\eta_j 1= 1001$, we find that $\alpha_j=1$ and $\beta_j=1$, meaning that the patterns of $\vec{\delta}$ on the left and the right sides of $1001$ arise from at least two faults, and the number of remaining faults is at most two. We can see that $\gamma_j=2$ because of the two zeros in the substring $1001$ and corresponds to three rounds with the same syndrome. Since the number of remaining faults that can cause the pattern $1001$ is less than the number of rounds with the same syndrome, the syndrome of at least one round in these three rounds must be correct and can be used for error correction.

		There are multiple, increasing fine-grained ways of estimating the number of faults in the context around each zero substring in the difference vector. Here, we use the term \emph{bit history} as a general term for a series of syndrome bits (measurement outcomes) from a given stabilizer generator, a given flag bit (the measurement outcome of a flag qubit), or bits in a difference vector (that is taken as the difference history of a group of bits). A key element in this discussion is the notion of correlated and uncorrelated bit histories.

		Under the assumed error model, two-bit histories are uncorrelated if they are independent of each other. For example, in our case, the circuit-level depolarizing channel is memoryless, and each fault can cause either one or two consecutive bits of ones. Thus, different sections of the same syndrome bit history that are at least two bits apart are uncorrelated as they are independent in time. Similarly, in space, if there are no shared qubits between two generators,  then their syndrome bit histories are completely independent. Also, flag qubits are always reset between rounds of measurements, and thus, all flag bits are  independent. However, when two stabilizer generators share at least one qubit, their syndrome bit histories are correlated. Similarly, due to hook errors, the flag qubit's bit history and the syndrome bit history of that same stabilizer generator are correlated.

		Our goal is to estimate the number of faults that occurred from a given bit history in the case of flag FTQEC. Estimates from uncorrelated histories can be summed together. When two or more estimates are from correlated histories, the best we can do is to take the maximum of those estimates. Note that the total estimates must not exceed the actual total number of occurred faults in any case, otherwise, the error correction protocol will not be fault tolerant.

		For the estimation in the previous work \cite{TPB23}, which is discussed previously in this section, the bits of the syndrome history before and after each substring $1 \eta_j 1$ are uncorrelated to the bits within $\eta_j$ under the memoryless depolarizing channel assumption. This means that $\alpha_j$ and $\beta_j$, which are the minimum numbers of faults that can cause the substring before and after $1 \eta_j 1$, can be independently estimated. The estimated number of faults in the context outside of the zero substring $\eta_j$ is, therefore, $\alpha_j+\beta_j$.

		In this work, we further extend the fault counting idea to flag FTQEC in which flag circuits with single flag qubit are used for syndrome extraction. Below, we will discuss two types of adaptive time decoders with different stop conditions, namely \emph{one-tailed and two-tailed adaptive time decoders}. Both protocols are applicable to any stabilizer code as long as flag circuits for the code that give a distinguishable fault set can be found. The flag FTQEC protocol with one-tailed adaptive time decoder satisfies the FTQEC conditions in \cref{def:FT_TL}, thus it is applicable to any fault-tolerant quantum computation as long as the fault-tolerant implementation of other operations (gate, state preparation, or measurement) also satisfies the revised definition of fault tolerance which considers the number of faults instead of the weight of the error \cite{TL22}. Meanwhile, the flag FTQEC protocol with two-tailed adaptive time decoder does not satisfy the FTQEC conditions in \cref{def:FT_TL} as the output error may correspond to a nontrivial cumulative flag vector, hence it is only applicable to quantum memory. Nevertheless, for a self-orthogonal CSS code, the FTQEC protocol with the two-tailed adaptive time decoder can be applied to any fault-tolerant Clifford computation if the cumulative flag vector is processed appropriately. An analysis of this extension will be discussed in \cref{subsec:ext_CSS}.

		\subsubsection{Two-tailed adaptive time decoder}
		\label{ssec:two-tailed}

		For the substring $1\eta_j 1$ in $\vec{\delta}$, suppose that bit one on the left of $\eta_j$ is the $i_1$-th bit of $\vec{\delta}$, and bit one on the right of $\eta_j$ is the $i_2$-th bit of $\vec{\delta}$. Let $\alpha_j,\beta_j,\gamma_j$ be defined as before, and let $\mu_j,\nu_j$ be the total numbers of nonzero flag bits obtained from round 1 to round $i_1$ and from round $i_2+1$ onward.
		Also, let $\omega_j$ be the sum of the numbers of flag bits that \emph{exceed 1 bit per round} during round $i_1+1$ to round $i_2$.
		For example, consider the substring $1\eta_j 1= 1001$ in the example below:
		\begin{align*}
			\text{Round:}        & \ 1 \ 2 \ 3 \ 4 \ 5 \ 6 \ 7 \ 8 \ 9 \ 10 \\
			\text{\# flag bits:} & \ 1 \ 0 \ 2 \ 0 | 0 \ 2 \ 1 | 0 \ 0 \ 1  \\
			\vec{\delta}\text{:} & \ \ 1 \ 1 \ 0 | 1 \ 0 \ 0 \ 1 | 0 \ 1
		\end{align*}
		In this example, $\alpha_j=1$, $\beta_j=1$, $\gamma_j=2$, $\mu_j=3$, and $\nu_j=1$, and $\omega_j=1$.

		Since a single fault can cause both nontrivial flag bits and syndrome differences (that is, syndrome bits and flag bits are correlated), one has to make sure that the number of faults is not overcounted. The numbers of faults that can cause bit histories before and after $1\eta_j 1$ are bounded from below by $\tilde{\alpha}_j=\max(\alpha_j,\mu_j)$ and $\tilde{\beta}_j=\max(\beta_j,\nu_j)$, respectively. So an estimate of the number of faults for the context outside of $\eta_j$ is $\tilde{\alpha}_j+\tilde{\beta}_j$.

		Next, let us consider $\eta_j$ of length $\gamma_j$ which corresponds $\gamma_j+1$ consecutive rounds with the same syndrome. To make all syndromes in this region incorrect, it requires at least one fault per round. So if we find a round with more than one flag bit, the number of flag bits that exceed one bit per round can be a part of the total estimate. That is, for each $\eta_j$, the total estimate is  $\tilde{\alpha}_j+\tilde{\beta}_j+\omega_j$.

		Under the assumption that there are at most $t$ faults in the whole protocol, if we find that there exists $\gamma_j$ such that $t-\tilde{\alpha_j}-\tilde{\beta_j}-\omega_j < \gamma_j+1$ (or equivalently, $\tilde{\alpha_j}+\tilde{\beta_j}+\gamma_j+\omega_j \geq t$), we know that a syndrome of at least one round in the $\gamma_j+1$ rounds that give rise to $\eta_j$ must be correct.

		Another way to find a correct syndrome is to estimate the total number of faults that can cause the whole syndrome and flag bit histories. Let $N_\mathrm{11}$ be the total number of non-overlapping $11$ sequences in the whole $\vec{\delta}$. Assuming that there are at most $t$ faults in the whole protocol, if $N_\mathrm{11} \geq t$, the last round must have a correct syndrome.

		Suppose that a table-based space decoder for flag FTQEC can be constructed. Then, a flag FTQEC protocol with two-tailed adaptive time decoder is as follows:

		\begin{protocol}{Flag FTQEC protocol with two-tailed adaptive time decoder}

			Let $t=\lfloor (d-1)/2\rfloor$ be the number of errors that a stabilizer code of distance $d$ can correct. Let $\vec{s}_i=(\vec{s}_{i,x},\vec{s}_{i,z})$ and $ \vec{F}_i=(\vec{F}_{i,x},\vec{F}_{i,z})$ be syndrome and cumulative flag vector obtained from the $i$-th round of full syndrome measurements with flag circuits. After the $i$-th round with $i\geq 2$, calculate $\delta_{i-1}$. Repeat syndrome measurements until one of the following conditions is satisfied, then perform error correction using the error syndrome corresponding to each condition:
			\begin{enumerate}
				\item For each $\eta_j$ in $\vec{\delta}$, calculate $\tilde{\alpha}_j,\tilde{\beta}_j,\gamma_j,\omega_j$. If at least one $\eta_j$ with $\tilde{\alpha_j}+\tilde{\beta_j}+\gamma_j+\omega_j \geq t$ is found, stop the syndrome measurements. Let $l$ be the last round of the $\gamma_j+1$ rounds that correspond to $\eta_j$. Perform $Z$-type error correction using $ (\vec{s}_{l,x},\vec{F}_{l-1,z})$, and perform $X$-type error correction using $ (\vec{s}_{l,z},\vec{F}_{l,x})$.
				\item Calculate $N_\mathrm{11}$ from the whole syndrome and flag bit histories. If $N_\mathrm{11} \geq t$, stop the syndrome measurements. Suppose that the latest round is round $l$. Perform $Z$-type error correction using $ (\vec{s}_{l,x},\vec{F}_{l-1,z})$, and perform $X$-type error correction using $ (\vec{s}_{l,z},\vec{F}_{l,x})$.
			\end{enumerate}
			\label{pro:flag_two-tailed}
		\end{protocol}

		The two-tailed adaptive time decoder for flag FTQEC developed in this work use similar ideas to the adaptive strong decoder presented in the previous work \cite{TPB23}. Therefore, the number of syndrome measurement rounds in the worst-case scenario is $(t+3)^2/4-1$ when $t$ is odd, and is $(t+2)(t+4)/4-1$ when $t$ is even. This can be proved by assuming that all faults does not cause any nonzero flag bits, then the rest of the proof follows the proof of Theorem 2 of the previous work \cite{TPB23}.

		If the syndrome $\vec{s}_l$ and cumulative flag vector $ \vec{F}_l=\sum_{i=1}^l \vec{f}_{i}\;(\mathrm{mod}\;2)$ of round $l$ are used for error correction, any faults that happened up to round $l$ will be corrected. However, because round $l$ may correspond to some $\eta_j$ in the middle of $\vec{\delta}$, an output error may correspond to a nontrivial cumulative flag vector. Therefore, \cref{pro:flag_two-tailed} may not satisfy FTQEC conditions in \cref{def:FT_TL} and cannot be applied to fault-tolerant quantum computation. Nevertheless, \cref{pro:flag_two-tailed} is still applicable to a quantum memory. To do so, one needs to pass the \emph{remaining cumulative flag vector} of the current FTQEC routine (the sum of the flag vectors from round $l+1$ onward) to the next FTQEC routine and use it as an initial flag vector.

		\subsubsection{One-tailed adaptive time decoder}
		\label{ssec:one-tailed}

		One-tailed and two-tailed decoders use similar ideas to estimate the number of faults, except that in the one-tailed case, the syndrome and cumulative vector for error correction must be from the very last zero substring in $\vec{\delta}$ (it is to ensure that the output error satisfies both conditions in \cref{def:FT_TL}). Suppose that $\vec{\delta}=\eta_1 1 \eta_2 1 \dots 1 \eta_c$ for some positive integer $c$, $\eta_c$ has length $\gamma_c \geq 1$, and bit one on the left of $\eta_c$ is the $i_1$-th bit of $\vec{\delta}$. We define $\alpha_c$ as usual and define $\mu_c$ to be the total number of nonzero flag bits obtained from round 1 to round $i_1$. Also, we define $\omega_c$ to be the sum of the numbers of flag bits that exceed 1 bit per round during round $i_1+1$ onward.
		Let $\tilde{\alpha}_c=\max(\alpha_c,\mu_c)$. In this case, the total estimate of the number of occurred faults is $\tilde{\alpha}_c+\omega_c$.

		Assuming that there are at most $t$ faults in the whole protocol, if we find that $\tilde{\alpha_c}+\gamma_c+\omega_c \geq t$, at least one round in the $\gamma_c+1$ rounds that give rise to $\eta_c$ must have a correct syndrome. This is the first possible stop condition.

		The second possible stop condition is similar to what we have for the two-tailed decoder. Let $N_\mathrm{11}$ be the total number of non-overlapping $11$ sequences in the whole $\vec{\delta}$. If $N_\mathrm{11} \geq t$, the last round must have a correct syndrome.

		Suppose that a table-based space decoder for flag FTQEC can be constructed. Then, a flag FTQEC protocol with the one-tailed adaptive time decoder is as follows:

		\begin{protocol}{Flag FTQEC protocol with one-tailed adaptive time decoder}

			Let $t=\lfloor (d-1)/2\rfloor$ be the number of errors that a stabilizer code of distance $d$ can correct. Let $\vec{s}_i=(\vec{s}_{i,x},\vec{s}_{i,z})$ and $ \vec{F_i}=(\vec{F}_{i,x},\vec{F}_{i,z})$ be syndrome and cumulative flag vector obtained from the $i$-th round of full syndrome measurements with flag circuits. After the $i$-th round with $i\geq 2$, calculate $\delta_{i-1}$. Repeat syndrome measurements until one of the following conditions is satisfied:
			\begin{enumerate}
				\item $\tilde{\alpha_c},\gamma_c,\omega_c$ satisfy $\tilde{\alpha_c}+\gamma_c+\omega_c \geq t$;
				\item $N_\mathrm{11} \geq t$.
			\end{enumerate}
			Suppose that the latest round when any condition is satisfied is round $l$. Perform $Z$-type error correction using $ (\vec{s}_{l,x},\vec{F}_{l-1,z})$, and perform $X$-type error correction using $(\vec{s}_{l,z},\vec{F}_{l,x})$.
			\label{pro:flag_one-tailed}
		\end{protocol}

		The number of rounds of full syndrome measurements in the worst-case scenario for \cref{pro:flag_one-tailed}, which is also the minimum number of rounds required to guarantee that error correction can be done, can be found by the following theorem:
		\begin{theorem}
			Suppose that flag circuits being used in \cref{pro:flag_one-tailed} give a distinguishable fault set $\mathcal{F}_t$, where $t=\lfloor (d-1)/2\rfloor$ and $d$ is the distance of the stabilizer code. Performing $\frac{t(t+3)}{2}+2$ rounds of full syndrome measurements is sufficient to guarantee that \cref{pro:flag_one-tailed} is strongly $t$-fault tolerant; i.e., both conditions in \cref{def:FT_TL} are satisfied.
			\label{thm:bound_flag}
		\end{theorem}

		\begin{proof}
			Suppose that $\vec{\delta}=\eta_1 1 \eta_2 1 \dots 1 \eta_c$ and $\gamma_c \geq 1$. We will show that if none of $\eta_1$, $\eta_1 1 \eta_2$, $\eta_1 1 \eta_2 1 \eta_3$, $\dots$, $\eta_1 1 \eta_2 1 \dots 1 \eta_c$ satisfies any condition in \cref{pro:flag_one-tailed}, the maximum length of such $\vec{\delta}$ is $\frac{t(t+3)}{2}$. In the worst-case scenario, flag measurement results do not help in estimating the number of occurred faults, so we can assume that $\tilde{\alpha}_c=\alpha_c$ and $\omega_c=0$. Below are the results from analyzing $\eta_1$, $\eta_1 1 \eta_2$, and $\eta_1 1 \eta_2 1 \eta_3$:
			\begin{enumerate}
				\item For $\eta_1$, $\alpha_c = 0$ and $\gamma_c=\gamma_1$, so the maximum length of $\eta_1$ such that $\tilde{\alpha_c}+\gamma_c \geq t$ is not satisfied is $t-1$.
				\item For $\eta_1 1 \eta_2$, $\alpha_c = 0$ and $\gamma_c=\gamma_2$, so the maximum length of $\eta_2$ such that $\tilde{\alpha_c}+\gamma_c \geq t$ is not satisfied is $t-1$.
				\item For $\eta_1 1 \eta_2 1 \eta_3$, $\alpha_c = 1$ and $\gamma_c=\gamma_3$, so the maximum length of $\eta_3$ such that $\tilde{\alpha_c}+\gamma_c \geq t$ is not satisfied is $t-2$.
			\end{enumerate}
			By induction, the maximum length of $\vec{\delta}=\eta_1 1 \eta_2 1 \dots 1 \eta_c$ such that $\tilde{\alpha_c}+\gamma_c \geq t$ is not satisfied is $(t-1)+1+(t-1)+1+(t-2)+1+\dots+1+1+0+1$, which is $\frac{t(t+3)}{2}$. Here $\vec{\delta}$ is of the form,
			\begin{equation}
				\underbrace{00\dots 00}_{t-1} 1 \underbrace{00\dots 00}_{t-1} 1 \underbrace{00\dots 00}_{t-2} 1 \underbrace{00\dots 00}_{t-3} 1\dots 1001011 \label{eq:max_delta_flag}
			\end{equation}

			The number of rounds that gives $\vec{\delta}$ of the maximum length is $\frac{t(t+3)}{2}+1$. By performing one more round of syndrome measurements, $\vec{\delta}$ is extended by one bit, which must be 0 if the total number of faults is no more than $t$. In that case, $\tilde{\alpha_c}+\gamma_c \geq t$ will be satisfied. Therefore, $\frac{t(t+3)}{2}+2$ rounds of full syndrome measurements is sufficient to guarantee that flag FTEC can be performed.

			Note that there are other forms of $\vec{\delta}$ in which none of $\eta_1$, $\eta_1 1 \eta_2$, $\eta_1 1 \eta_2 1 \eta_3$, $\dots$, $\eta_1 1 \eta_2 1 \dots 1 \eta_c$ satisfies any condition in \cref{pro:flag_one-tailed}, and the length of $\vec{\delta}$ is $\frac{t(t+3)}{2}-1$; For example, suppose that $t=3$. Possible forms of such $\vec{\delta}$ are $001101011$, and $001001111$. In any case, one of the conditions in \cref{pro:flag_one-tailed} will be satisfied if one more round of syndrome measurements is done, so the number of rounds to guarantee fault tolerance is still $\frac{t(t+3)}{2}+2$.
		\end{proof}

		Note that the number given by \cref{thm:bound_flag} is worse than that of the two-tailed decoder because we are not allowed to check whether the syndrome of any round in the middle can be used for error correction.

		An advantage of the FTQEC protocol with one-tailed adaptive time decoder is that it is applicable to any kind of fault-tolerant quantum computation as long as the corresponding fault-tolerant implementation satisfies the revised definitions of fault tolerance which consider the number of faults instead of the weight of errors \cite{TL22}. This is possible because when the syndrome and cumulative flag vector for error correction are from the last zero substring in $\vec{\delta}$, it is guaranteed that the output error corresponds to a zero cumulative flag vector.

		\subsection{Extended techniques for CSS codes}
		\label{subsec:ext_CSS}

		In this section, we discuss two additional techniques which can further improve our flag FTQEC protocols with adaptive time decoding. The first technique is the separated $X$ and $Z$ counting which is applicable to any CSS code. This technique is based on the ideas from Refs. \cite{DR20,TPB23}, and can be used to improve the pseudothreshold. The main difference from the technique developed in Ref. \cite{TPB23} is that this work also uses flag information to estimate the number of occurred faults, making the procedure to obtain a syndrome for error correction terminate faster. The second technique is the classical processing of the remaining cumulative flag vector. This technique allows our flag FTQEC protocol with the two-tailed adaptive time decoder to be applicable to any fault-tolerant Clifford computation.

		\begin{table*}[tbph!]
			\begin{center}
				\begin{tabular}{| c | c | c |}
					\hline
					Remaining cumulative flag vector                             & Logical Clifford               & Initial flag vector of the                                                                            \\
					of the current FTQEC routine                                 & operation                      & next FTQEC routine                                                                                    \\
					\hline
					$ (\vec{F}_x,\vec{F}_z)$                                     & $\bar{H}$                      & $ (\vec{F}_z,\vec{F}_x)$                                                                              \\
					\hline
					$ (\vec{F}_x,\vec{F}_z)$                                     & $\bar{S}$                      & $ (\vec{F}_x,\vec{F}_x \oplus \vec{F}_z)$                                                             \\
					\hline
					$ (\vec{F}_{x,1},\vec{F}_{z,1}|\vec{F}_{x,2},\vec{F}_{z,2})$ & $\overline{\text{CNOT}}_{1,2}$ & $ (\vec{F}_{x,1},\vec{F}_{z,1}\oplus \vec{F}_{z,2}|\vec{F}_{x,1} \oplus \vec{F}_{x,2},\vec{F}_{z,2})$ \\
					\hline
				\end{tabular}
			\end{center}
			\caption{A list of required classical processing operations on the remaining cumulative flag vector in case that a logical Clifford gate is performed between two FTQEC routines. With these operations, a flag FTQEC protocol with two-tailed adaptive time decoder or separated $X$ and $Z$ counting is applicable to any fault-tolerant Clifford computation, given that the CSS code is self-orthogonal.}
			\label{tab:flag_proc}
		\end{table*}

		\subsubsection{Separated X and Z counting}
		\label{ssec: separate xz}

		For any CSS code, $Z$-type and $X$-type errors can be corrected separately. It is possible to improve the number of measurements by separating the $X$-type and $Z$-type syndrome measurement rounds (which correspond to $X$-type and $Z$-type stabilizer generators). In this section, we introduce the XZ and ZX decoding strategies. In the XZ strategy, first, we execute a time decoder (which can be Shor, one-tailed, or two-tailed decoder) using only the $X$-type syndromes. The difference vector for this process is denoted by $\vec{\delta}_x$. After the decoder returns the $X$-type syndrome and the cumulative flag vectors for $Z$-type error correction, we estimate the number of faults $t_x$ that could cause $\vec{\delta}_x$; we define $\alpha_{\mathrm{all},x}$ to be the total number of non-overlapping $11$ sequences plus the total number of remaining 1s in $\vec{\delta}_x$, define $\mu_{\mathrm{all},x}$ to be the total number of nontrivial flag bits in $\vec{\delta}_x$, and let $t_x=\max(\alpha_{\mathrm{all},x},\mu_{\mathrm{all},x})$. Given that we spend this number of faults from our fault budget $t$, we can reduce the target number of faults in the stop condition for the $Z$-type syndrome measurements. Afterward, we run a time decoder for $Z$-type syndromes with the target number of faults $t_z = t-t_x$. The ZX strategy is similar to the XZ strategy, except that the $Z$-type generators are measured first.

		When the separated $X$ and $Z$ counting technique is applied to a flag FTQEC protocol, one can find syndromes for $Z$-type and $X$-type error corrections faster compared to a conventional method where the target numbers of faults for both types of error corrections are $t$. However, a drawback is that the flag FTQEC protocol will only be compatible with quantum memory. This is because of each type of error correction requires flag information of the opposite type. In particular, suppose that the time decoder for $X$-type syndrome measurements give syndrome $\vec{s}_x$ and cumulative flag vector $ \vec{F}_x$, and the time decoder for $Z$-type syndrome measurements give syndrome $\vec{s}_z$ and cumulative flag vector $ \vec{F}_z$. $Z$-type error correction will be done by applying a space decoder to $\vec{s}_x$ and the zero cumulative flag vector, while $X$-type error correction will be done by applying a space decoder to $\vec{s}_x$ and $ \vec{F}_z$. The cumulative flag vector $ \vec{F}_x$ which has not been used will be treated as the remaining cumulative flag vector of the current FTQEC routine and used as an initial flag vector for $Z$-type error correction in the next FTQEC routine.

		\subsubsection{Classical processing of the remaining cumulative flag vector}
		\label{ssec:remaining_flag_proc}

		One drawback of a flag FTQEC protocol that uses the two-tailed adaptive time decoder or the separated $X$ and $Z$ counting technique is that it is only applicable to a quantum memory, not a general fault-tolerant quantum computation. This is because the output error at the end of each FTQEC routine may correspond to a nontrivial cumulative flag vector. To correct such an error, one needs to pass the flag information from each FTQEC routine (the remaining cumulative flag vector) to the next FTQEC routine. However, if there is some quantum computation between two FTQEC routines (as in an extended rectangle \cite{AGP06}), the error will be transformed and may not be correctable if the corresponding flag information is not processed properly.

		Nevertheless, for any self-orthogonal CSS code, a flag FTQEC protocol with two-tailed adaptive time decoder or separated $X$ and $Z$ counting (or both) can made applicable to any fault-tolerant Clifford computation. For example, let us consider an application of a logical Hadamard gate $\bar{H}$ between two FTQEC routines. Suppose that the first FTQEC routine causes an output error $E_x\cdot E_z$ and the remaining cumulative flag vector is $ (\vec{F}_x,\vec{F}_z)$. Without a logical Hadamard gate, $E_x$ and $E_z$ can be corrected using $ \vec{F}_z$ and $ \vec{F}_x$, respectively. A logical Hadamard gate transforms an $X$-type error to a $Z$-type error of the same form, and vice versa. Because the $X$-type and $Z$-type generators are of the same form, possible fault combinations for both types of errors are also of the same form. To correct the transformed error $\bar{H}(E_x\cdot E_z)\bar{H}^\dagger$ in the second FTQEC routine, one needs to swap the $X$-type and $Z$-type cumulative flag vector; that is, the initial flag vector for the second FTQEC routine must be $ (\vec{F}_z,\vec{F}_x)$.

		We can apply similar ideas for flag information processing to logical $S$ and logical CNOT gates. The summary of the classical processing operations for logical $H$, $S$, and $\text{CNOT}$ gates is provided in \cref{tab:flag_proc}. Because $\{H,S,\text{CNOT}\}$ generates the Clifford group, a flag FTQEC protocol with two-tailed adaptive time decoder or separated $X$ and $Z$ counting is applicable to any fault-tolerant Clifford computation given that the CSS code is self-orthogonal. Note that the magic state distillation and injection \cite{BK05,BH12} use only Clifford operations. Thus, our techniques are also applicable to fault-tolerant universal quantum computation given that high-fidelity magic states are provided.

		\section{Numerical results}
		\label{sec:numerical results}

		\subsection{Methods} \label{subsec:methods}

		Our optimization tools for space and time decoders including the compact lookup table construction, the MIM technique, and the adaptive time decoders for flag FTQEC are applicable to any stabilizer code. However, we focus on a specific family of codes where the aforementioned tools can be simplified and extended techniques, including separated $X$ and $Z$ decoding and classical processing of flag information are applicable---the family of self-orthogonal CSS codes in which the number of physical qubits is odd, the number of logical qubits is 1, and logical $X$ and $Z$ operators are transversal. To evaluate the performance of our tools, we simulate FTQEC protocols on the \codepar{(3d^2 + 1)/4,1,d} hexagonal color codes \cite{BM06} of distance 3, 5, 7 and 9. These codes are planar topological codes with configurations displayed in \cref{fig: tcc}. For each code, stabilizer generators are measured using the syndrome extraction circuits with single flag ancilla, as depicted in \cref{fig:flag_w}. It was proven that for the hexagonal code of any distance, using flag circuits of this form preserves the code distance regardless of the gate orderings \cite{CKYZ20,TL22}). The simulation is implemented under the circuit-level depolarizing noise model specified in \cref{subsec:FTQEC_def}. As there is no idling noise in our error model, the syndromes can be extracted sequentially.

		To construct a lookup table for space decoding and to verify that our circuit configurations preserve the code distance, we implement the algorithm described in \cref{subsection:compact lookup} using C++. The timing for verification alongside the statistics of the lookup table can be found in \cref{tab: decoder numbers}. The lookup table for these codes can be generated on the fly before the sampling starts as the required time is low enough.

		Here we simulate the storage (i.e. the result of the logical identity operation) of the logical state $\ket{\bar{0}}$. We use of the Pauli frame simulator in Stim \cite{Gidney21} to collect measurement samples, and use Cirq \cite{Cirq22} for constructing the circuits with the given noise model. After a perfect preparation of $\ket{\bar{0}}$, we perform noisy error correction and recovery. In the error correction process, full rounds of syndrome measurements are repeated until the stop condition of the time decoder is satisfied. The time decoder returns an accepted full syndrome (consisting of error syndrome and cumulative flag vector), then the space decoder determines the recovery operation based on the accepted full syndrome. This recovery operation is applied to the data qubits afterwards. Finally, we apply an ideal error correction and determine whether the output error is a logical $X$ error (which corresponds to having $\ket{\bar{1}}$ as the output state).

		\begin{figure}[tbp]
			\centering
			\includegraphics[width=0.45\textwidth]{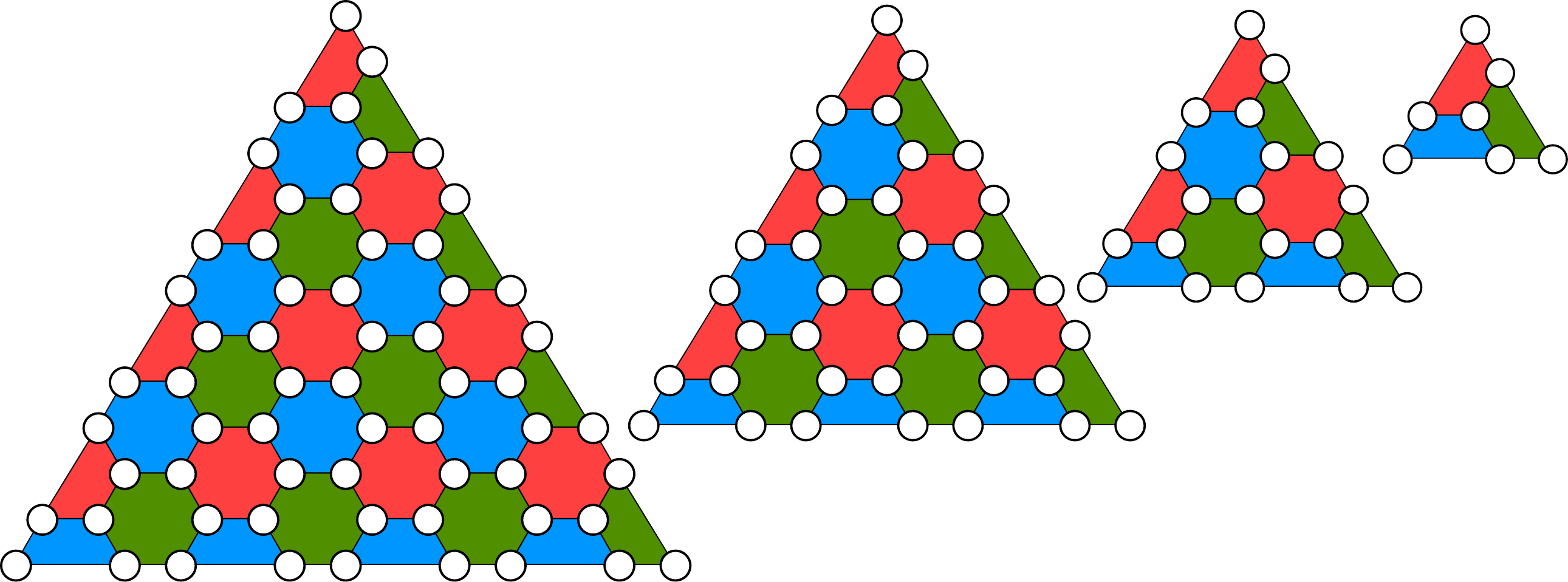}
			\caption{The studied members of the hexagonal color code, for distances 3, 5, 7, and 9 (right to left). Qubits are on the vertices and stabilizer generators are the plaquettes. As the codes are self-orthogonal CSS codes, both the $X$ and $Z$ stabilizer generators are described by the same layout.}
			\label{fig: tcc}
		\end{figure}

		\subsection{The overall effect of optimization tools}

		We first compare two protocols: (1) the FTQEC protocol with Shor time decoder without the MIM technique (the protocol in which none of our optimization tools are applied) and (2) the FTQEC protocol with the MIM technique and the two-tailed adaptive time decoder with the ZX strategy (the best FTQEC protocol in this work which is compatible with any Clifford computation on a self-orthogonal CSS code). The logical error rate $p_L$ vs physical error rate $p$ for hexagonal color codes of distance 3, 5, 7, and 9 are plotted in \cref{fig: threshold}. Our results show that for each code, applying the optimization tools can significantly improve the pseudothreshold (the intersection between each plot and the $p_L=2p/3$ line). Furthermore, the optimized decoder yields orders of magnitude improvements in the logical error rate in the $p=10^{-4}$ error regime.

		\begin{figure}[tbph]
			\centering
			\includegraphics[width=0.47\textwidth]{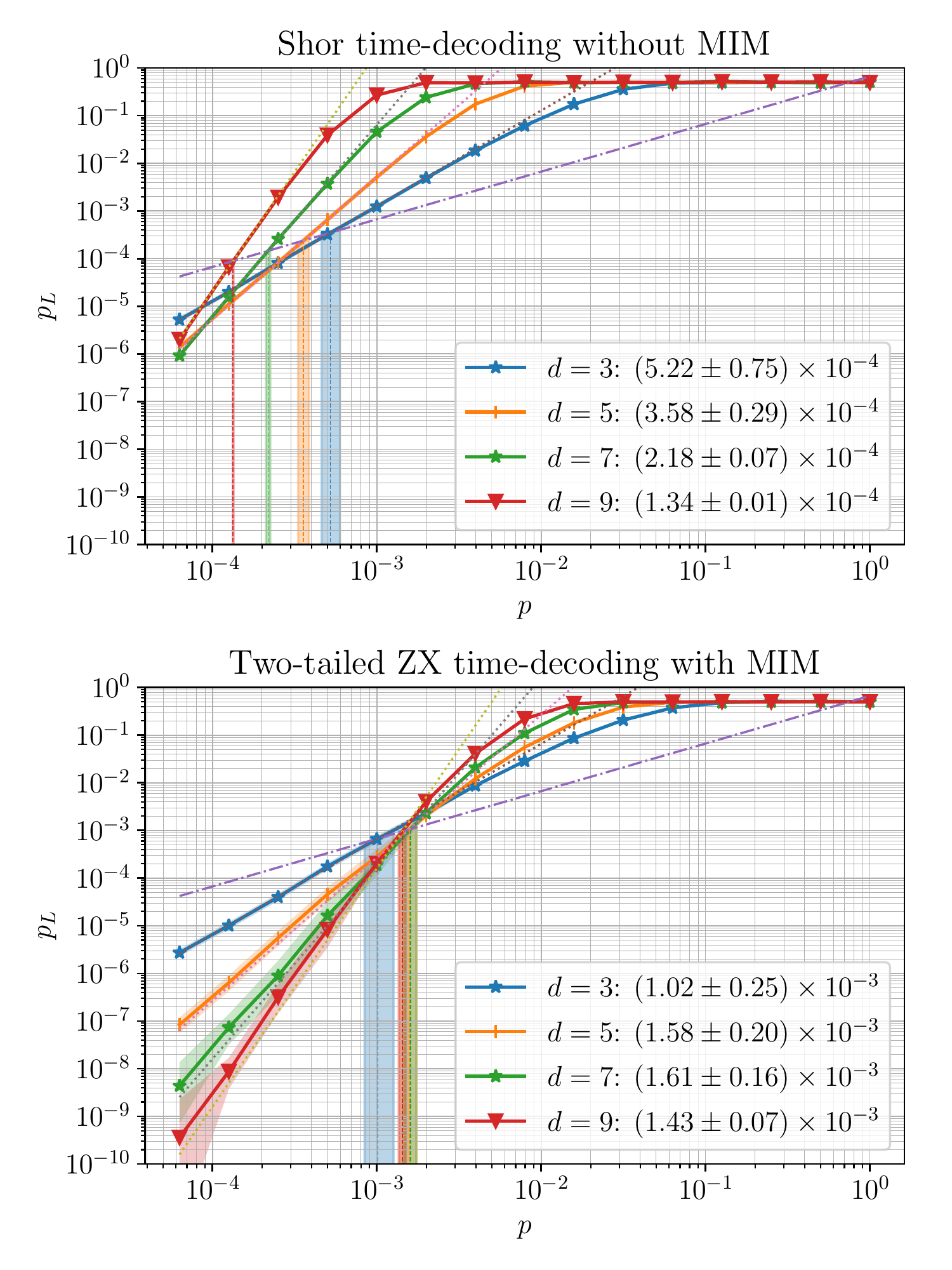}
			\caption{The upper plot shows the curve of logical error rate $p_L$ vs physical error rate $p$ for the hexagonal color code family without any of our optimization techniques, using the Shor time decoder without MIM. The lower plot uses the best-performing combination of our techniques, including MIM and the two-tailed adaptive time decoder with ZX strategy. Pseudothresholds for each curve (the $p_{th}$ error rate which gives $p_L(p_{th})=2p_{th}/3$) are included in the labels and marked with vertical lines. The data points represent the number of logical errors divided by the total number of samples at that $p$ error rate and thus estimate the true logical error rates, which should lie within the shaded areas with high confidence. The dotted helper lines, which are $\alpha p^{t+1}$ where $\alpha=\frac{2}{3}  p_{th}^{-t}$ retroactively calculated for each curve from its pseudothreshold, show good agreement with distance preservation.}
			\label{fig: threshold}
		\end{figure}

		Under a noise model parameterized by a single parameter $p$, the fault-tolerant threshold $p_{th}$ is the error probability under which the logical error rate is guaranteed to decrease with increasing code distance for a specific code family and decoder. Our decoders can yield a $p_{th}$ for concatenated code families using a level-by-level decoder, but they will not yield a threshold for topological code families for two reasons. The practical reason is that our space decoder that uses a lookup table is not scalable to the large $d$ limit. The fundamental reason is that the time decoder will always take $\vec{\delta}$ in which all bits are one when $d$ is large, because $\delta_j$ for each round will be 1 with a probability exponentially close to 1 for finite $p$. The space decoder then acts on the final state but lacks the information about correlations to properly correct it. This is why an efficient space-time decoder is critical for achieving $p_{th}$ for topological codes.

		We can define an effective threshold $\tilde{p}_{th}$ as the error rate below which increasing the code distance improves the logical error rate for this finite set of codes. The optimized protocol yields a $\tilde{p}_{th}=1.5\times10^{-3}$, while the unoptimized protocol yields $\tilde{p}_{th}=4.5\times10^{-5}$. We also note that the crossing point between the codes of distances $d$ and $d-2$ is dropping quickly with the unoptimized decoder, while it is stable for the optimized decoder over this code set. \cref{tab:d9-effects-summary} summarizes the effects of different optimization tools on the pseudothreshold of the $d=9$ color code. In the next sections, we further discuss the effect of each technique that can contribute to this improvement.

		\begin{table}[tbph]
			\centering
			\begin{tabular}{|l|c|c|}
				\hline
				Time decoder  & w/o MIM ($\times 10^{-4}$) & MIM ($\times 10^{-4}$) \\
				\hline
				\hline
				Shor          & $1.34 \pm 0.01$            & $2.79 \pm 0.01$        \\
				One-tailed    & $2.11 \pm 0.05$            & $3.91 \pm 0.26$        \\
				Two-tailed    & $3.38 \pm 0.17$            & $6.30 \pm 0.45$        \\
				Two-tailed XZ & --                         & $6.09 \pm 0.47$        \\
				Two-tailed ZX & --                         & $14.3 \pm 0.7$         \\
				\hline
			\end{tabular}

			\caption{The effect of different time decoders (rows) and the MIM technique (columns) on the pseudothreshold of the distance 9 hexagonal color code. See \cref{fig: mim rates,fig: sfa,fig: zx} for more details on the underlying data.}
			\label{tab:d9-effects-summary}
		\end{table}

		\subsection{The effect of the Meet-in-the-Middle technique} \label{subsec:MIM_effect}

		In this section, we evaluate the performance of simulated storage that uses the space decoder with and without the MIM technique. We explore the effect for distances 3, 5, 7, and 9, and compare the effect when the time decoder is Shor, one-tail, or two-tail time decoder. We observe a significant decrease in logical error rates and an improvement in pseudothreshold when the MIM technique is applied. We also find that the benefit increases with the code distance. In \cref{fig: mim rates}, we show the improvement for the code of distance 9 where the benefit is the largest. The results for codes of other distances are provided in \cref{fig: app mim shor rates,fig: app mim 1t rates,fig: app mim 2t rates}.

		It is clear that both non-adaptive (Shor) and adaptive time decoders benefit from the MIM technique. The pseudothreshold for the Shor time decoder increases by more than 100\%, from $(1.34 \pm 0.01) \times 10^{-4}$ to $(2.79 \pm 0.01) \times 10^{-4}$. The pseudothreshold for one-tailed adaptive time decoder increases by 85\%, from $(2.11 \pm 0.05) \times 10^{-4}$ to $(3.91 \pm 0.26) \times 10^{-4}$. Finally, the pseudothreshold for the two-tailed adaptive time decoder gets a boost of 76\%, from $(3.38 \pm 0.17) \times 10^{-4}$ to $(5.96 \pm 0.71) \times 10^{-4}$.

		\begin{figure}[tbph]
			\centering
			\includegraphics[width=0.47\textwidth]{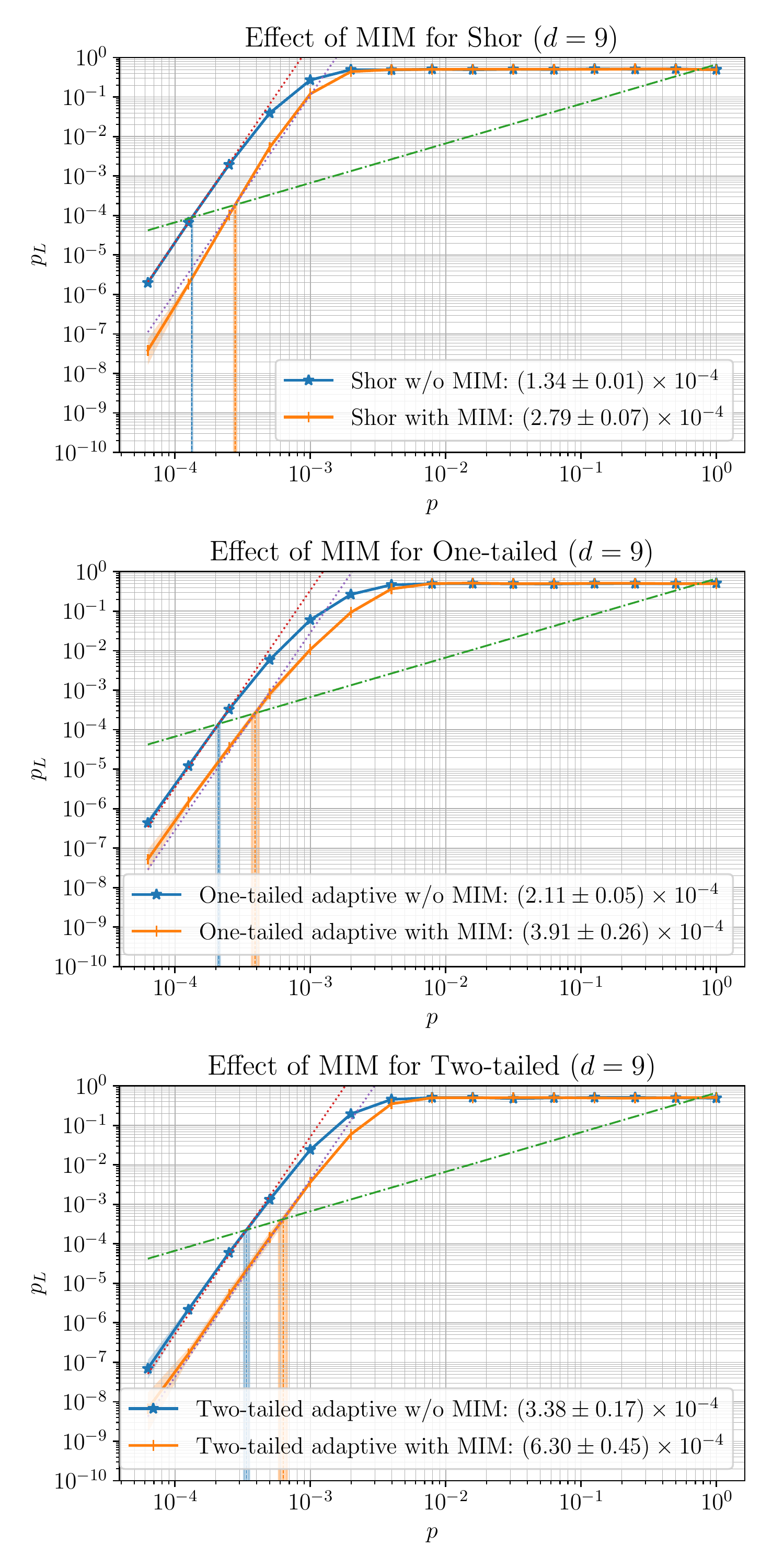}
			\caption{The effect of the MIM technique on different time decoders at distance 9. The effect is the largest for Shor time decoder, more than doubling the pseudothreshold. The MIM technique also gives at least a significant 70\% improvement on the adaptive time decoders.}
			\label{fig: mim rates}
		\end{figure}

		\subsection{The effect of the adaptive time decoders} \label{subsec:adaptive_effect}

		In this section, we compare the performance of the simulated storage numerical experiments that use different time decoders when the MIM technique is applied. The results are displayed in \cref{fig: sfa} for the hexagonal code of distance 9, and we refer the reader to \cref{fig: app sfa} in \cref{app: all} for the results for the codes of other distances.

		For the code of distance 9, in comparison with the Shor time decoder, the one-tailed adaptive time decoder improves the pseudothreshold by 40\% from $(2.79 \pm 0.07) \times 10^{-4}$ to $(3.91 \pm 0.07) \times 10^{-4}$. The two-tailed method achieves $(5.96 \pm 0.71) \times 10^{-4}$ pseudothreshold, which is more than a 100\% increase compared to the Shor time decoder. However, this gain vanishes at lower error rates, and the performances of Shor and one-tailed decoders become similar at around $p=10^{-4}$. It is not surprising as we expect all adaptive time decoders to converge to Shor time decoder at lower error rates. The main reason for this convergence is that the performance gains for the adaptive techniques come from a decrease in the average number of rounds for syndrome measurements, and the decrease converges to zero at low error rates. How fast the decrease converges does matter, and in contrast to the one-tailed approach, the two-tailed time decoder preserves its performance gain over Shor time decoder at the observed low-error regime as low as $5 \times 10^{-5}$.

		We also provide the plots of the average numbers of full rounds of measurements for all decoders. At a low-error-rate regime, all decoders have the same minimum number of measurement rounds, $t+1$, which corresponds to the case that all bits in the difference vector are zeros. We can see the separation more clearly when the physical error rate is in the $10^{-3}$ range; the two-tailed time decoder requires the fewest rounds, followed by the one-tailed decoder, and the Shor time decoder performs the worst. At the high-error-rate regime, all bits in the difference vector tend to be ones. In this case, the Shor time decoder requires $(t+1)^2$ rounds, while both one-tailed and two-tailed decoders require $2t+1$ rounds.

		\begin{figure*}[tbph]
			\centering
			\includegraphics[width=\textwidth]{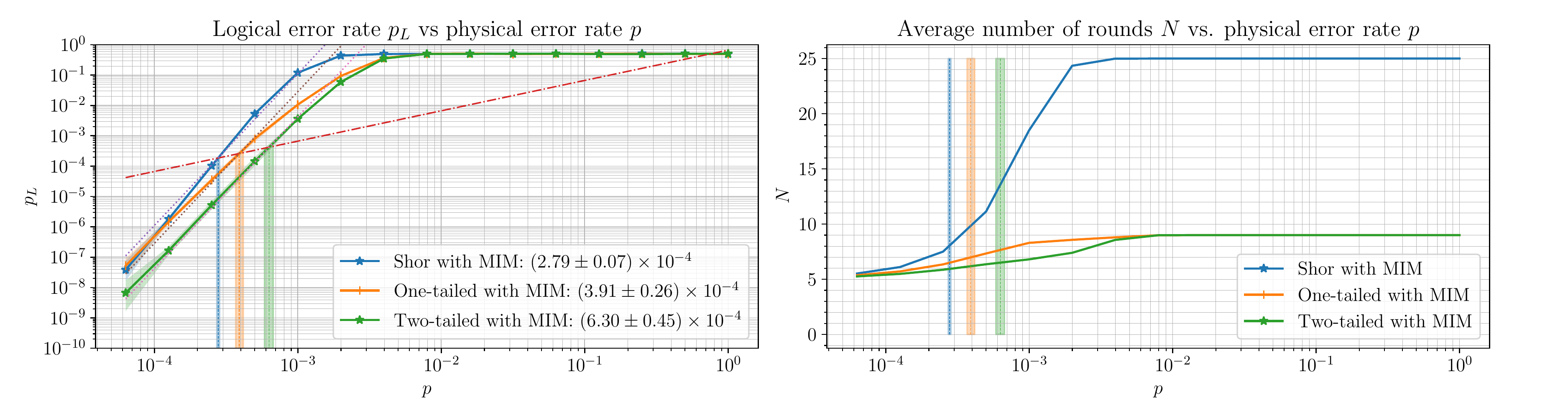}
			\caption{Logical error rates of one-tailed and two-tailed adaptive time decoders compared to the Shor time decoder (left) with corresponding average number of rounds (right) for the hexagonal color code of distance 9.}
			\label{fig: sfa}
		\end{figure*}

		\subsection{The effect of the separated X and Z counting technique} \label{subsec:XZ_effect}

		\begin{figure*}[tbph]
			\centering
			\includegraphics[width=\textwidth]{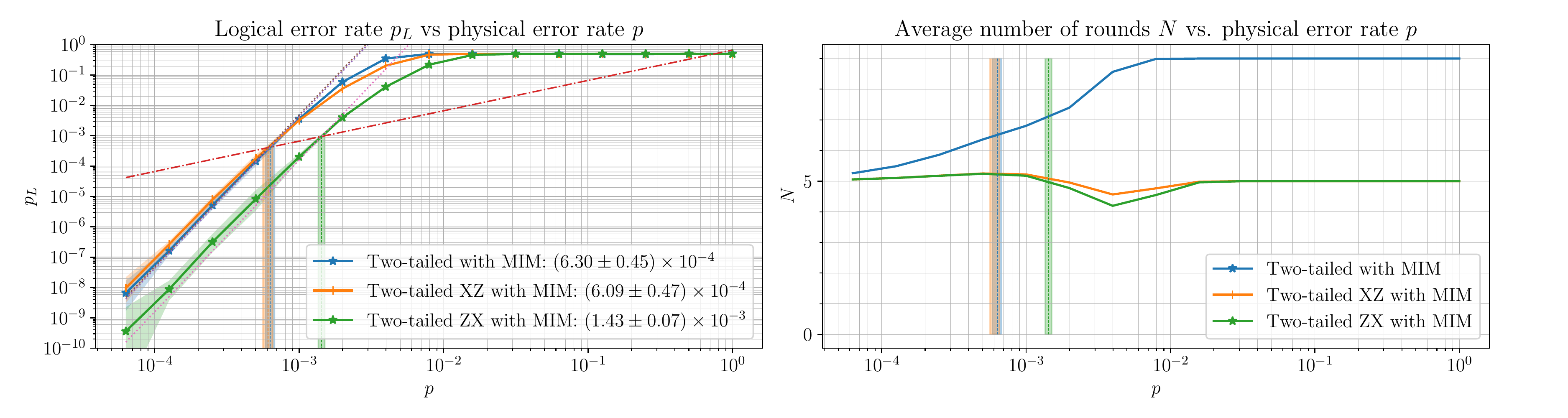}
			\caption{Logical error rates of the two-tailed time decoder with XZ and ZX strategies in comparison with the two-tailed adaptive time decoder with joint $X$ and $Z$ measurements (left) and corresponding average number of rounds (right) for the hexagonal color code of distance 9.}
			\label{fig: zx}
		\end{figure*}

		In this section, we observe the performance gains when the separated $X$ and $Z$ counting technique is applied. Here we compare the FTQEC protocols that use the two-tailed adaptive time decoder with joint $X$ and $Z$ generator measurements (as in \cref{ssec:two-tailed}), the two-tailed adaptive time decoder with XZ strategy, and the two-tailed adaptive time decoder with ZX strategy (as in \cref{ssec: separate xz}). The logical error rate is calculated from the number of samples in which the output error is a logical $X$ error. The $p_L$ versus $p$ plots for the code of distance 9 are shown in \cref{fig: zx} (the results for codes of other distances can be found in \cref{fig: app zx} in \cref{app: all}).

		In terms of the pseudothreshold, we observe that the decoder with separated $X$ and $Z$ counting performs the best when $Z$-type generators are measured before $X$-type generators. Compared to the two-tailed decoder with joint $X$ and $Z$ generator measurements, the separated two-tailed ZX decoder improves the pseudothreshold by 140\% from $(5.96 \pm 0.71) \times 10^{-4}$ to $(1.44 \pm 0.20) \times 10^{-3}$. This is mainly because measuring generators of the first type ($X$ or $Z$) requires more rounds, and it is more probable that the measurements can cause correlated errors of the same type as the generators being measured (which are more difficult to correct than uncorrelated errors since they require flag information). Because in our simulations we measure the performance of storing the logical $\ket{0}$ state (thus, a logical $X$ error is counted), the decoder that measures $X$-type generators first performs worse. We also observe that there is no significant difference between the two-tailed decoder with joint measurements and the two-tailed decoder with XZ strategy.

		We also provide plots of the average number of full rounds of measurements for all decoders (where the full round of single-type generator measurements is counted as half a round of total measurements). At the low-error-rate regime, all decoders require $t+1$ rounds. For the original two-tailed decoder, the average number of rounds increases as the physical error rate increases, and it reaches $2t+1$ rounds at the high-error-rate regime. For both two-tailed decoders with separated $X$ and $Z$ counting, we find that the average number of rounds increases near the pseudothreshold, then there are the dips after the pseudothreshold, and the numbers reach $t+1$ rounds at the high-error-rate regime. The dips come from the fact that the measurements of generators of the first type (either $X$ or $Z$) can stop at less than $(2t+1)/2$ rounds but the estimate of the number of occurred faults can be $t$, which then causes the measurements of generators of the second type to stop at $1/2$ rounds. At the high-error-rate regime, the decoders with separated $X$ and $Z$ counting require $t+1$ rounds since measuring generators of the first type requires $(2t+1)/2$ rounds while measuring generators of the second type requires $1/2$ rounds on average. Overall, the decoder that measures $Z$-type generators first performs better than the decoder that measures $X$-type generators first.

		\section{Discussions and conclusions} \label{sec:discussions}

		In this work, we focus on flag FTQEC with lookup table decoding and improvements to a decoder consisting of a time decoder and a space decoder. For the space decoder, we first develop a technique to build the lookup table more efficiently in \cref{subsection:compact lookup}. With our lookup table construction method, the lookup table for a self-orthogonal CSS code requires at least 87.5\% less memory compared to the lookup table for a generic stabilizer code. The construction method also verifies the distinguishability of the fault set corresponding to flag circuits for syndrome measurements. Our construction also leads to the notion of the fault code, a linear code corresponding to the faults under circuit-level noise, which simplifies the verification of the distance of the protocol. More efficient decoding schemes for the fault code can be an interesting avenue to explore in future work.

		Another optimization tool for space decoding is the MIM technique in \cref{subsec:MIM}, which could improve decoding accuracy when the number of faults in the protocol is greater than $t$ (where $t=\lfloor (d-1)/2 \rfloor$ for the code of distance $d$). The effect of the MIM technique on the simulated storage of the hexagonal color codes is discussed in \cref{subsec:MIM_effect} (see also \cref{fig: app mim shor rates}). We find that for any kind of time decoder, the logical error rates are reduced, and the pseudothresholds are improved when applying the MIM technique, with greater improvements at larger distances.

		For the time decoder, we generalize the adaptive syndrome measurement technique from the previous work \cite{TPB23} (which is applicable to Shor-style error correction \cite{Shor96}) to flag FTQEC, and develop one-tailed and two-tailed adaptive time decoders in \cref{subsec:adaptive_time_decoder}. For a general stabilizer code in which flag FTQEC is possible, the one-tailed decoder is preferable as it is compatible with any fault-tolerant quantum computation, while the two-tailed decoder is applicable to quantum memory only. Nevertheless, for self-orthogonal CSS codes, the two-tailed decoder is applicable to any fault-tolerant computation built from Clifford gates and application of $T$ gates by gate teleportation using high-fidelity magic states with the help of the classical processing technique on cumulative flag vectors developed in \cref{subsec:ext_CSS}. The effect of the adaptive time decoders on the simulated storage is discussed in \cref{subsec:adaptive_effect}. We observe that our adaptive time decoders can improve the pseudothresholds compared to the non-adaptive (Shor) time decoder while preserving the code distance. The two-tailed decoder also outperforms the one-tailed decoder.

		The two-tailed adaptive decoder without MIM in this work is similar to the adaptive strong decoder in the previous work \cite{TPB23}, except that this work uses flag circuits instead of syndrome extraction circuits with cat states. The numerical results show that using flag circuits results in a 20-35\% increase of the pseudothreshold for the hexagonal color codes of distances 3, 5, 7 and 9. This is mainly because flag circuits have fewer state preparation and qubit measurement locations, although they have more gates. The previous work \cite{TPB23} also assumes fault-tolerant preparation of cat states, which requires verification \cite{Shor96} or ancilla decoding circuit \cite{DA07} that can result in higher space and time overhead. Thus, the pseudothresholds could be worse in that case if additional requirements are also considered. It should be noted that flag circuits may not outperform syndrome extraction circuits with cat states in general, as flag FTQEC for other codes may require more complicated flag circuits.

		We can further improve the performance of adaptive time decoders on self-orthogonal CSS codes by using the separated $X$ and $Z$ counting technique described in \cref{subsec:ext_CSS}. Here, we estimate the number of faults occurred from the measurement of generators of the first type (either $X$ or $Z$) and then use that information in the measurement of generators of the second type. The effect of this technique can be found in \cref{subsec:XZ_effect}. When the logical $\ket{0}$ state is stored, we find that the protocol that measures $Z$-type generators before $X$-type generators performs the best. We see no significant difference in the protocol that measures $X$-type generators before $Z$-type generators, and the protocol that measures $X$-type and $Z$-type generators jointly. Thus, the separated $X$ and $Z$ counting provides an advantage only for certain input states depending on the measurement order.

		Combining all techniques together, we find a significant improvement in the pseudothreshold while the code distance is still preserved. For example, on the hexagonal color code of distance 9, the pseudothreshold goes up from $(1.34\pm 0.01) \times 10^{-4}$ to $(1.42 \pm 0.12) \times 10^{-3}$. We also find that in comparison with the unoptimized decoder, the crossing points between the codes of distances $d$ and $d-2$ come much closer when all techniques are applied (as shown in \cref{fig: threshold}), leading to a higher effective threshold $\tilde{p}_{th}$ for this set of codes.

		While our techniques are applicable to a broader family of codes, it would be interesting to see how our results compare with other works that study error decoding on the hexagonal color codes under circuit-level noise. For example, Baireuther et al. \cite{BCC19} reported a pseudothreshold above $2 \times 10^{-3}$ (against $p_L=p$ instead of $p_L=2p/3$) with a neural-network decoder, which also preserves the code distance empirically. However, it was also reported that training decoders for $d>7$ became too expensive. By adapting efficient color-decoding algorithms known as restriction decoder \cite{KD19} and projection decoder \cite{Delfosse_2014}, Chamberland et al. \cite{CKYZ20} and Beverland, Kubica and Svore \cite{BKS21} reported threshold values of $2\times 10^{-3}$ and $3.7 \times 10^{-3}$ respectively. The difference between threshold values is mostly contributed by different choices of syndrome extraction circuits: for each weight-six stabilizer generator, Ref. \cite{CKYZ20} used three flag qubits for connectivity considerations, while Ref. \cite{BKS21} did not use any flag qubits. However, both the restriction decoder and the projection decoder can only correct up to $d/3$ errors (see Fig. 15 in Sahay and Brown \cite{SB22} for example failure modes) on the color code family considered in this paper\footnote{See Appendix A of the 3rd arXiv version of \cite{CKYZ20}.}. Recent preprints report distance-losing schemes to decode the color code with even higher thresholds of $4.7 \times 10^{-3}$ \cite{ZWG23}, and between $5\times 10^{-3}$ to $7\times 10^{-3}$ \cite{GJ23} without using flag qubits.

		In contrast to the constructions that utilize the restriction decoder \cite{CKYZ20} and the projection decoder \cite{BKS21}, our adaptive decoding method preserves the code distance (although the lookup table is not scalable to codes with larger distances). It is expected that our method could become advantageous for the codes of interest when the physical error rate is below a certain value. However, the noise models in Refs. \cite{BCC19,CKYZ20,BKS21,ZWG23,GJ23} also consider idling noise, while our noise model does not. Sequential syndrome extraction is expected to perform poorly in architectures where idling noise is dominant (see \cref{app: idling noise} for an analysis on the \codepar{7,1,3} code). To improve performance, our methods need to be combined with optimized schedules specific to the given code family. CNOT schedule optimization is involved, requiring an enumeration of valid CNOT schedules satisfying basic constraints and finding the best-performing one using exhaustive search, similar to how Beverland, Kubica, and Svore\cite{BKS21} found a well-performing schedule for hexagonal color codes and bare ancillas. It is thus an open question what the error regime is where our flag qubit-based, adaptive methods are advantageous in comparison to the non-distance preserving decoders. This analysis will require evaluation using code-specific optimizations under different strength idling noise scenarios, which we leave for future work.

		Hierarchical decoding approaches also provide an interesting avenue to explore with lookup table-based and adaptive techniques \cite{D20,SBB23}. We conjecture that our techniques may result in efficient pre-decoders. The lookup tables and the adaptive syndrome algorithms would have to be restricted to local sections of topological codes or sparsely connected modules of other codes. Then, when the lookup table decoders cannot decode the local problems, the more expensive and accurate decoder can attempt to decode the nonlocal problem.

		It should be noted that this work uses the adaptive syndrome measurement technique, which assumes fast qubit preparation and measurement. For the architectures on which qubit measurement and reset are slow, however, our method may require a large number of ancillas or may not be possible. In that case, one may consider using the flag schemes that do not require fast qubit measurement and reset, such as the flag scheme for any distance-3 code \cite{PR21}, or the flag scheme in which the flag gadgets are constructed from the classical BCH codes \cite{AM22}.

		\section{Acknowledgements}

		The work was supported by the Office of the Director of National Intelligence - Intelligence Advanced Research Projects Activity through an Army Research Office contract (W911NF-16-1-0082), the Army Research Office Multidisciplinary University Research Initiative (MURI) program (W911NF-16-1-0349), the Army Research Office (W911NF-21-1-0005), and the National Science Foundation Institute for Robust Quantum Simulation (QLCI grant OMA-2120757).

		\bibliographystyle{ieeetr}
		\bibliography{bibtex_FT_Duke}

		\clearpage
		\onecolumngrid
		\appendix

		\color{defaultcolor}
		\section{Memory footprint savings in the lookup table}\label{app:memory footprint}

		In this section, we detail how much savings each of the ideas in the main text contributes relative to the lookup table memory cost of
		\begin{align}
			M_{stab}=T_{stab}(4n-2k) \text{ bits, } \label{eq:mstab}
		\end{align}
		when we look at a code as a generic stabilizer code. If CROs and the logical class are used instead of full Pauli operators for recovery, then instead of storing the full $2n$ bits for recovery, only $2k$ bits are required. Thus, $M_{stab,CRO}=T_{stab}(2n)$ leading to:
		\begin{align}
			M_{stab,CRO}/M_{stab}=\frac{T_{stab}(2n)}{T_{stab}(4n-2k)} = \frac{1}{2-k/n} = \frac{1}{2-R},
		\end{align}
		where $R=k/n$, the encoding rate. Thus, for codes with $R\rightarrow 0$ as $n \rightarrow \infty$, this alone saves up to 50\% in storage.

		For CSS codes, as mentioned in the main text, we assume independent recovery for $X$ and $Z$ type errors. Let $T_X$ and $T_Z$ denote the number of unique $X$- and $Z$-type non-trivial syndromes (which are obtained by measuring $X$- and $Z$-type generators, respectively). Then, CSS codes have two lookup tables mapping pure $X$- and $Z$-type syndromes to purely $Z$- and $X$-type recovery operators, with $T_X+1$ and $T_Z+1$ entries. Thus, they require only $2r_X$ and $2r_Z$ bits for the map key (factor of two for flags and data bits) and $n$ bits for the recovery operator. This results in the memory cost for the lookup tables of a CSS code:
		\begin{align}
			M_{CSS} & = (T_X+1)(2 r_X + n) + (T_Z+1)(2 r_Z + n)  \text{ bits }                  \nonumber   \\
			        & = 2 r_X T_X +2 r_X + nT_X + n + 2 r_Z T_Z +2 r_Z + nT_Z + n   \text{ bits } \nonumber \\
			        & = 2 (r_X T_X + r_Z T_Z ) + 2 (r_X+r_Z) + n(T_X + T_Z) + 2n   \text{ bits }  \nonumber \\
			        & = 2(r_Z T_Z +r_X T_X)   + (4n-2k) + n(T_Z + T_X) \text{ bits }
		\end{align}

		To compare this with $M_{stab}$, we need to take care of the trivial syndrome and introduce a variable, $T_{XZ}$, for the number of unique mixed $X$/$Z$ syndromes that a generic stabilizer code representation would yield:
		\begin{align}
			T_{stab} = T_X + T_Z + 1 + T_{XZ}
		\end{align}
		Thus, from \cref{eq:mstab}:
		\begin{align}
			M_{stab} = (T_X + T_Z + 1 + T_{XZ})(4n-2k) \text{ bits }
		\end{align}

		The ratio between the storage cost for a CSS code versus the cost when the same code is viewed as a generic stabilizer code is hard to bound precisely. Nevertheless, at least we know that $M_{stab} > M_{CSS}$ as the savings are,
		\begin{align}
			M_{stab} - 	M_{CSS} & = (T_X + T_Z + 1 + T_{XZ})(4n-2k) - (T_X+1)(2 r_X + n) + (T_Z+1)(2 r_Z + n) \nonumber \\
			                   & = 2(T_Xr_Z + T_Zr_X) + n(T_X + T_Z) + 2 T_{XZ}(2n-k) \text{ bits. }
		\end{align}

		If we use CROs, we can reduce the size of the map values to $k$ bits from $n$:
		\begin{align}
			M_{CSS,CRO} & = (T_X+1)(2 r_X + k) + (T_Z+1)(2 r_Z + k) \text{ bits }                    \nonumber  \\
			            & = 2 r_X T_X +2 r_X + kT_X + k + 2 r_Z T_Z +2 r_Z + kT_Z + k   \text{ bits } \nonumber \\
			            & = 2 (r_X T_X + r_Z T_Z ) + 2 (r_X+r_Z) + k(T_X + T_Z) + 2k   \text{ bits } \nonumber  \\
			            & = 2(r_X T_X +r_Z T_Z)   + 2n + k(T_Z + T_X)  \text{ bits }
		\end{align}
		And thus, not surprisingly, the decrease in bits is:
		\begin{align}
			M_{CSS} - M_{CSS,CRO} = (n-k)(2 + T_Z + T_X) = (n-k) T_{CSS}
		\end{align}
		where $T_{CSS}$ is the total number of entries of the two tables.

		If the code is self-orthogonal, then the two tables coincide, $T:=T_X=T_Z$, $r:=r_X=r_Z$. Thus,
		\begin{align}
			M_{CSS,CRO,SO} & = \frac{1}{2} \left( 2(r_Z T_Z +r_X T_X)   + 2n + k(T_Z + T_X) \right) \nonumber \\
			               & =\frac{1}{2} \left( 2(r T +r T)   + 2n + k(T + T ) \right)            \nonumber  \\
			               & =\frac{1}{2} \left( 4rT   + 2n + 2kT \right)                         \nonumber   \\
			               & = 2rT   + n + kT                                                      \nonumber  \\
			               & = (n-k)T   + n + kT                                                   \nonumber  \\
			               & = nT   + n                                                           \nonumber   \\
			               & = n(T+1),
		\end{align}
		which is consistent with having $T+1$ entries, with a map key of $n-k$ bits (with $(n-k)/2$ for flags and for generator bits) and a value with $k$ bits. If we would view the self-orthogonal CSS code as a generic stabilizer code, we would get $T_{stab}=1+2T+T_{XZ}$ and thus the ratio is,
		\begin{align}
			M_{CSS,CRO,SO}/M_{stab} & = \frac{n(T+1)}{T_{stab}(4n-2k)}     \nonumber  \\
			                        & = \frac{n(T+1)}{(1+2T+T_{XZ})(4n-2k)} \nonumber \\
			                        & = \frac{T+1}{(1+2T+T_{XZ})(4-2R)}    \nonumber  \\
			                        & \leq \frac{1}{8-4R},
		\end{align}
		where we used the fact that $T_{XZ}$ must be at least 1 for $t>0$ and a non-trivial encoding. This upper bound means that at a zero rate code, leveraging the structure of a self-orthogonal CSS code and the CROs can create a memory footprint less than 12.5\% that of the memory footprint of a lookup table if we viewed the code as a stabilizer code.

		\color{defaultcolor}
		\section{The effect of idling noise} \label{app: idling noise}

		To demonstrate the effect of idling noise, we evaluate the \codepar{7,1,3} code under a naive interleaved schedule, depicted in \cref{sub@subfig:INT_noiseless} without noise terms and in \cref{sub@subfig:INT_WI2} with gate noise terms with strength $p=0.02$ and idling noise terms with strength $p_I=0.01$. Note that further improvements are possible to reduce idling in the circuit by doubling the number of flag qubits and ancilla qubits and measuring $X$ and $Z$ stabilizer generators in parallel, similar to the scheme by Beverland, Kubica, and Svore \cite{BKS21}. This will, however, be only possible for the two-tailed adaptive decoder, and the separated $X$/$Z$ decoder will not work by definition. Also, protocol-specific CNOT schedule optimization might be possible depending on the underlying quantum code. As we are not aiming to find tools on the code level, this investigation is out of scope of this paper. It is also interesting to point out that using a single flag qubit and single ancilla forces sequential execution of the gates within a generator, while multi-flag based schemes such as in the work of Chamberland et al. \cite{CKYZ20} allow for multiple CNOT gates to be executed in the same time step. While our methods here use single flag qubit-based protocols, that angle can be relaxed if the strength of idling noise requires it.

		Our numerical evaluation results displayed in \cref{fig: app idling} show that at idling noise strength $p_I=p$, the pseudothreshold is 20 to 25 times smaller than the case without idling noise $p_I=0$. However, as the relative strength of the depolarizing noise $p/p_I$ increases, the performance approaches the ideal case rapidly. Furthermore, we can see that our decoders still preserve the distance, which is expected given that the single qubit depolarizing noise terms do not change the set of errors to be corrected but only change the strength of some terms.

\begin{figure*}[htpb!]
	\centering
	\includegraphics[width=.8\textwidth]{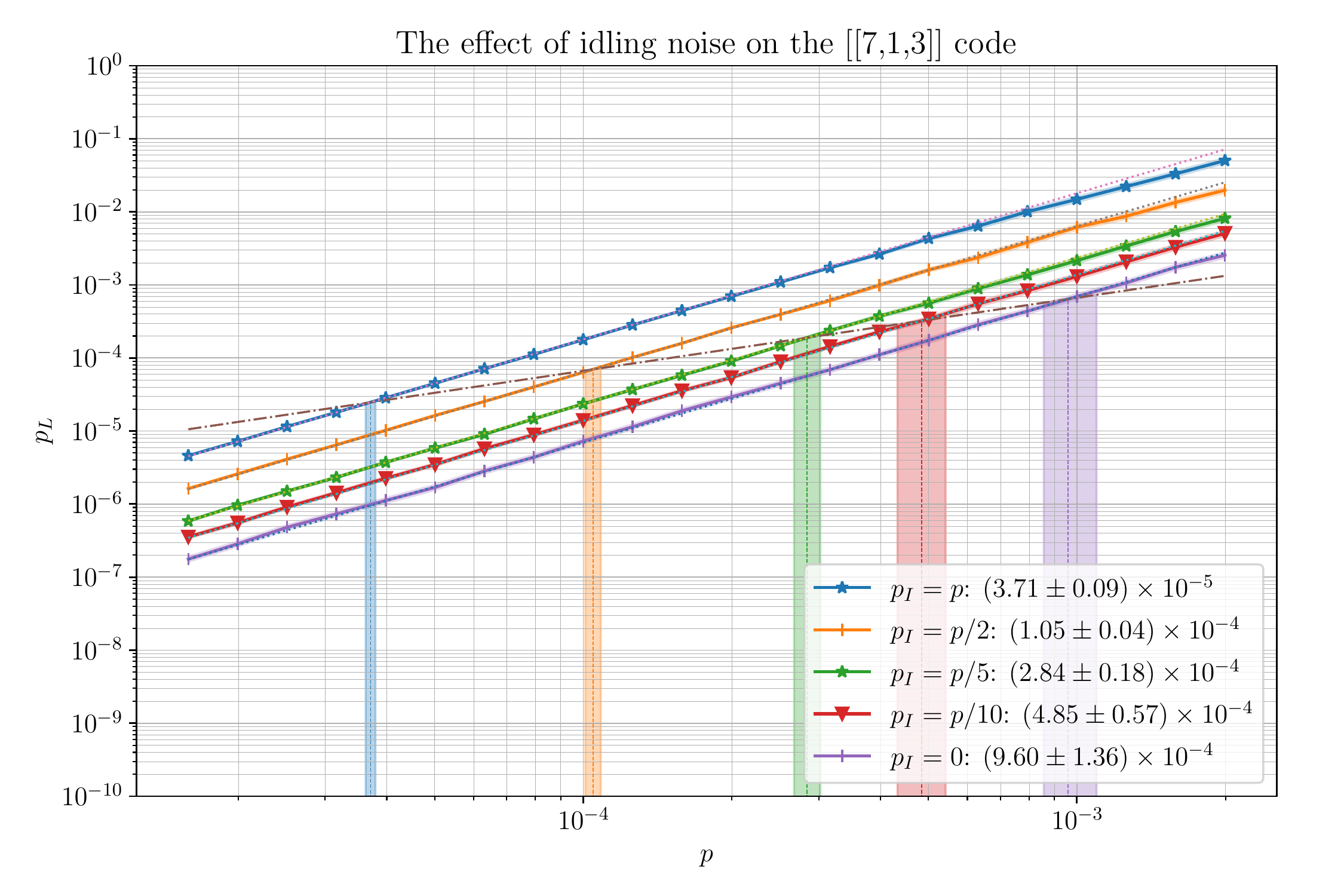}
	\caption{The effect of idling noise on a naive CNOT schedule for the \codepar{7,1,3} code at different idling noise strength $p_I$ relative to the gate errors $p$. In this setup $p_I=p$ is the full, standard depolarizing noise model, and $p_I=0$ is the one we used to evaluate our methods in the main text, while $p_I=p/2, p_I=p/5$ and $p_I=p/10$ are between those two extremes.}
	\label{fig: app idling}
\end{figure*}

\begin{figure*}
	\begin{subfigure}[b]{1\textwidth}
		\includegraphics[width=.5\textwidth]{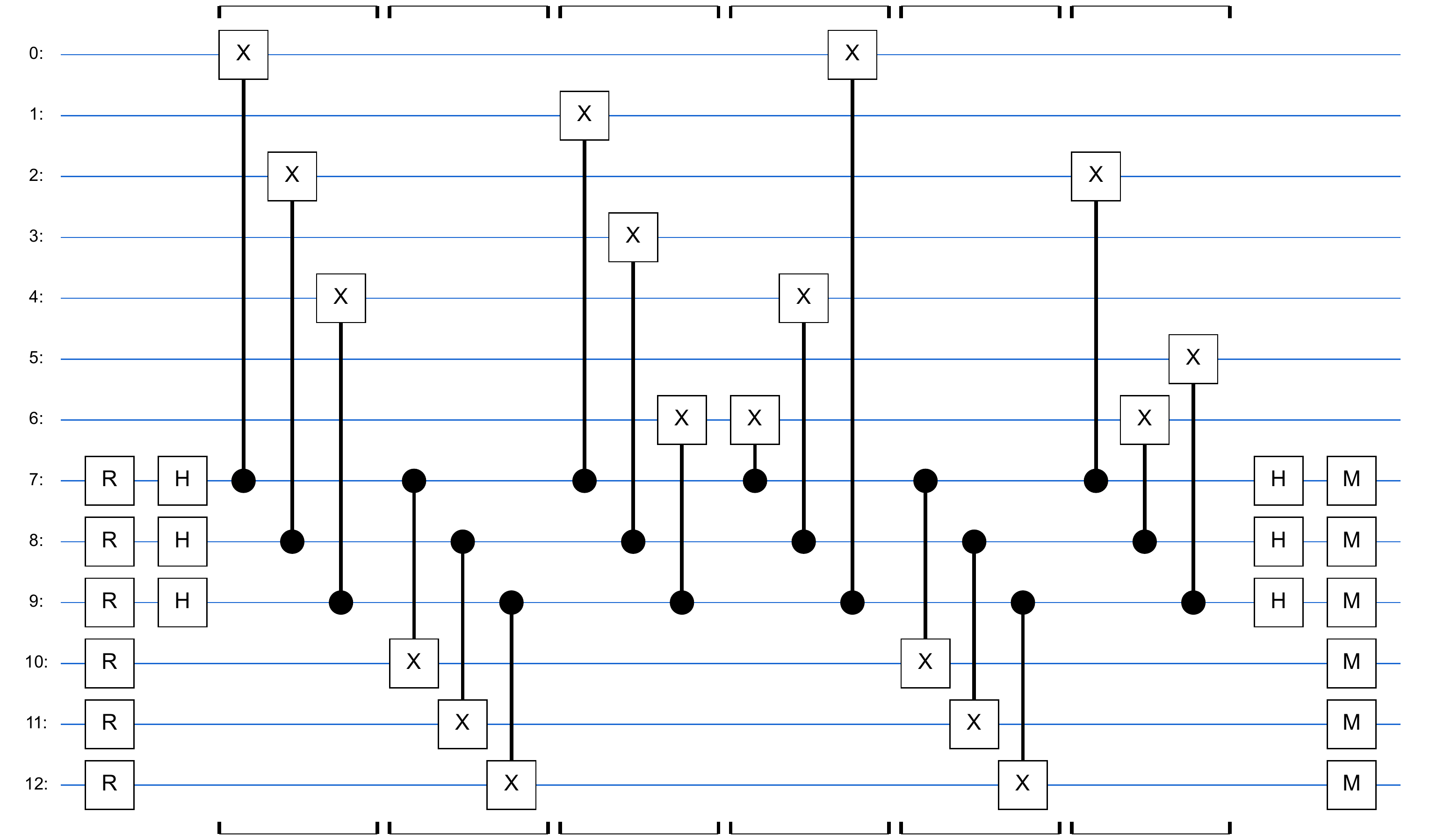}
		\captionsetup{justification=centering}
		\caption{}
		\label{subfig:INT_noiseless}
	\end{subfigure}
	\begin{subfigure}[b]{1\textwidth}
		\includegraphics[width=1\textwidth]{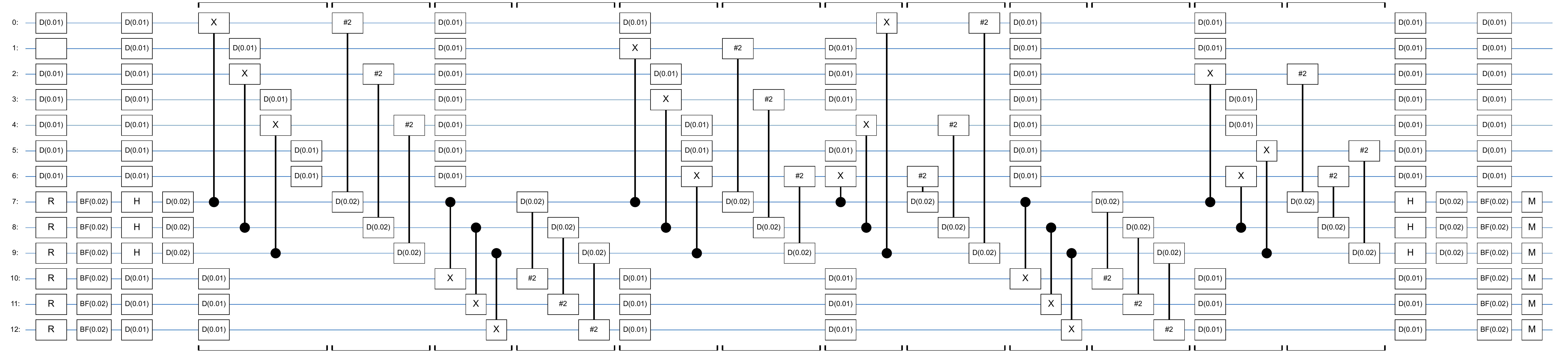}
		\captionsetup{justification=centering}
		\caption{}
		\label{subfig:INT_WI2}
	\end{subfigure}
	\caption{An interleaved schedule of extracting the $Z$ syndrome of the \codepar{7,1,3} code without (a) and with (b) noise terms at gate depolarizing strength $p=0.02$ and idling noise strength $p_I=p/2=0.01$. Data qubits are 0 to 6, ancilla qubits are 7 to 9 and flag qubits are 10 to 11. Brackets above and below the circuit group gates together that are executed during the same time step. $X$-type syndrome extraction is similar.}
	\label{fig:INT}
\end{figure*}

\pagebreak
\color{black}

\section{Figures for all distances}\label{app: all}

\begin{figure}[h!]
	\centering
	\includegraphics[width=1\textwidth]{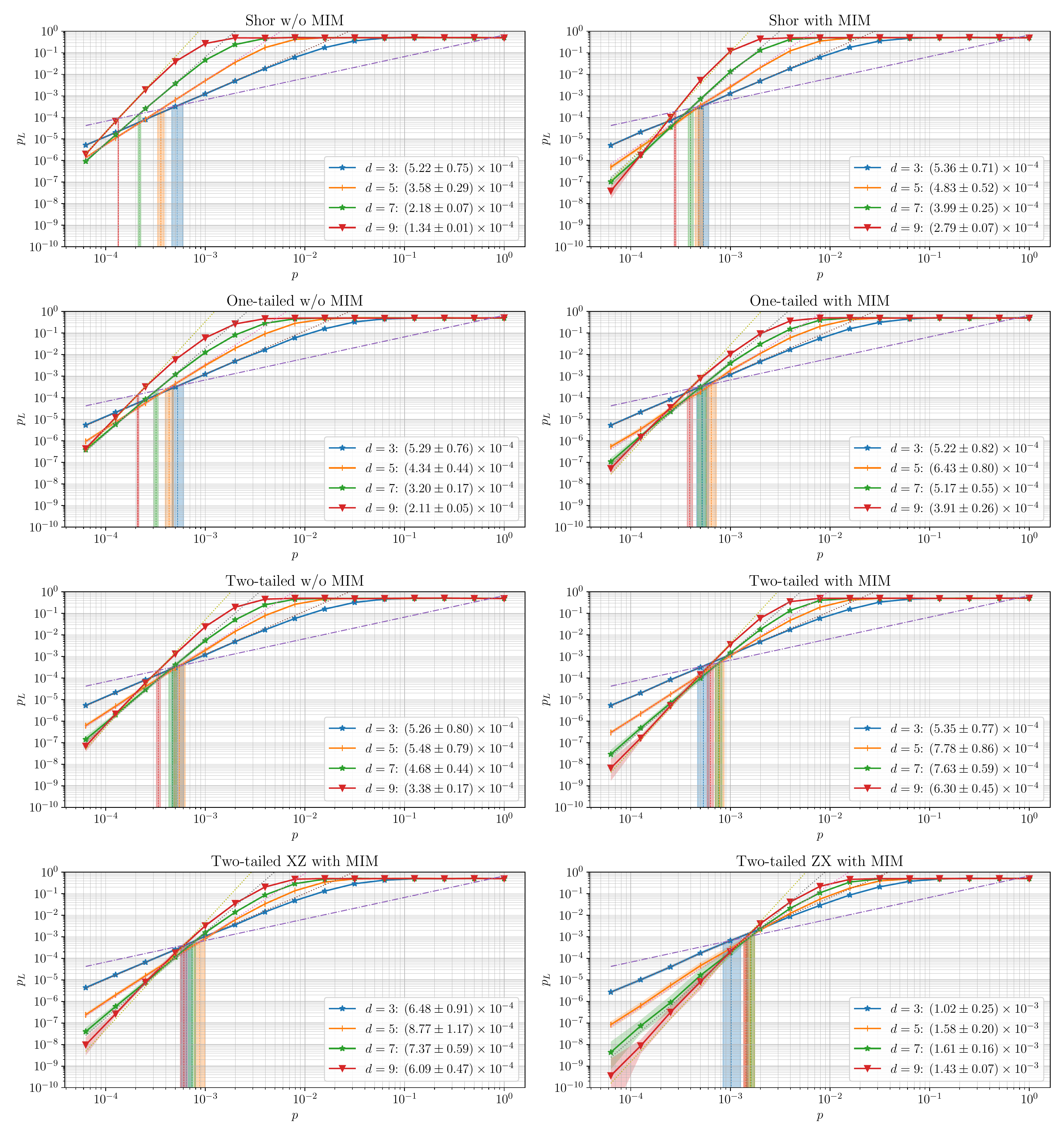}
	\caption{The threshold formation effect of increasingly better space and time decoders. Both space decoding improvement (MIM) and time decoding improvements (from Shor to two-tailed ZX-strategy) help in making the intersections of the $p_L$ vs $p$ curves more focused.}
	\label{fig: app thresholds}
\end{figure}

\begin{figure*}[tbph]
	\centering
	\includegraphics[width=1.0\textwidth]{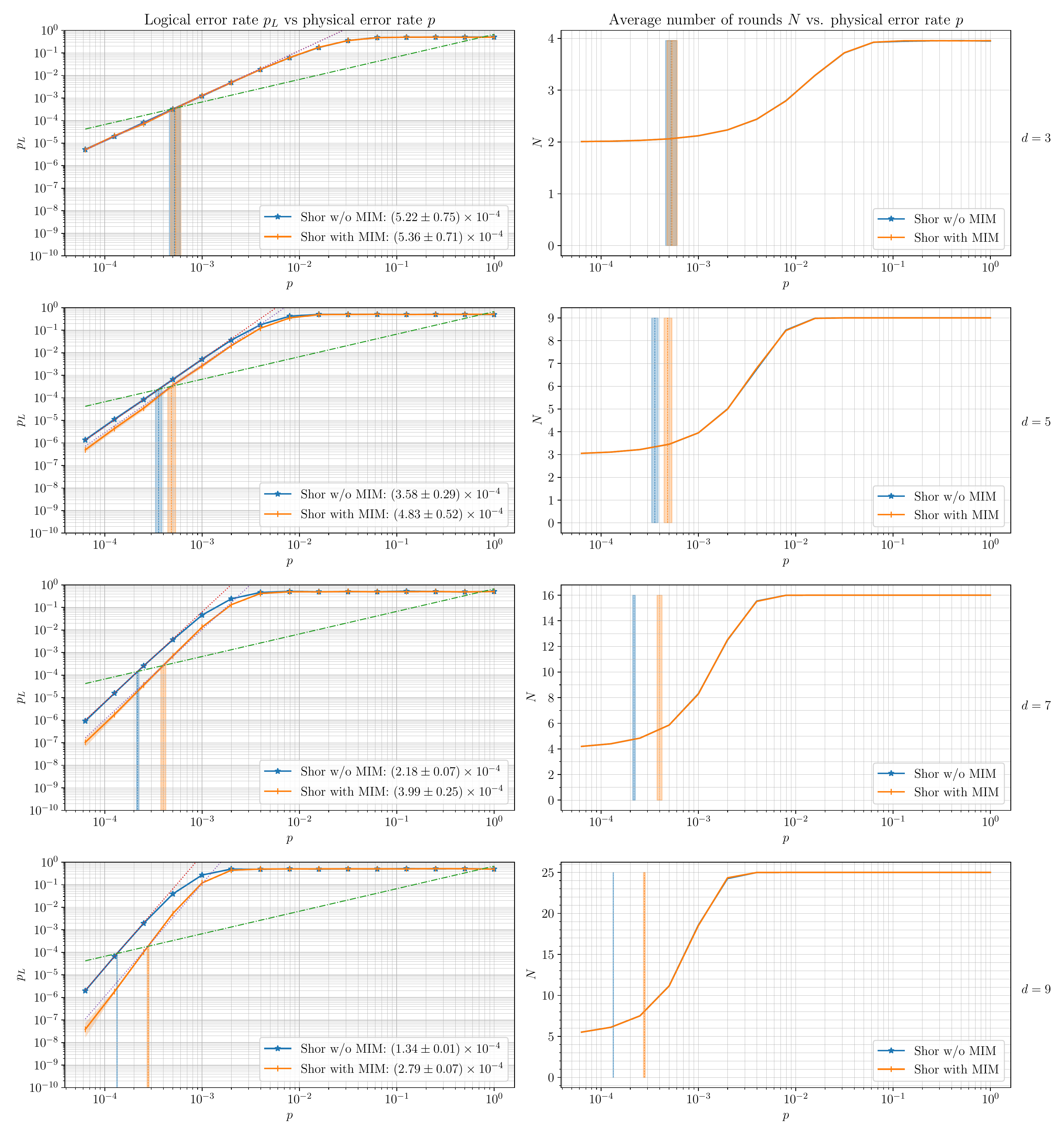}
	\caption{The effect of the MIM technique on Shor time decoder for hexagonal color codes of distances 3, 5, 7, and 9. The improvement is increasing with the code distance, with no improvement at $d=3$ and the biggest one at $d=9$.}
	\label{fig: app mim shor rates}
\end{figure*}

\begin{figure*}[tbph]
	\centering
	\includegraphics[width=1.0\textwidth]{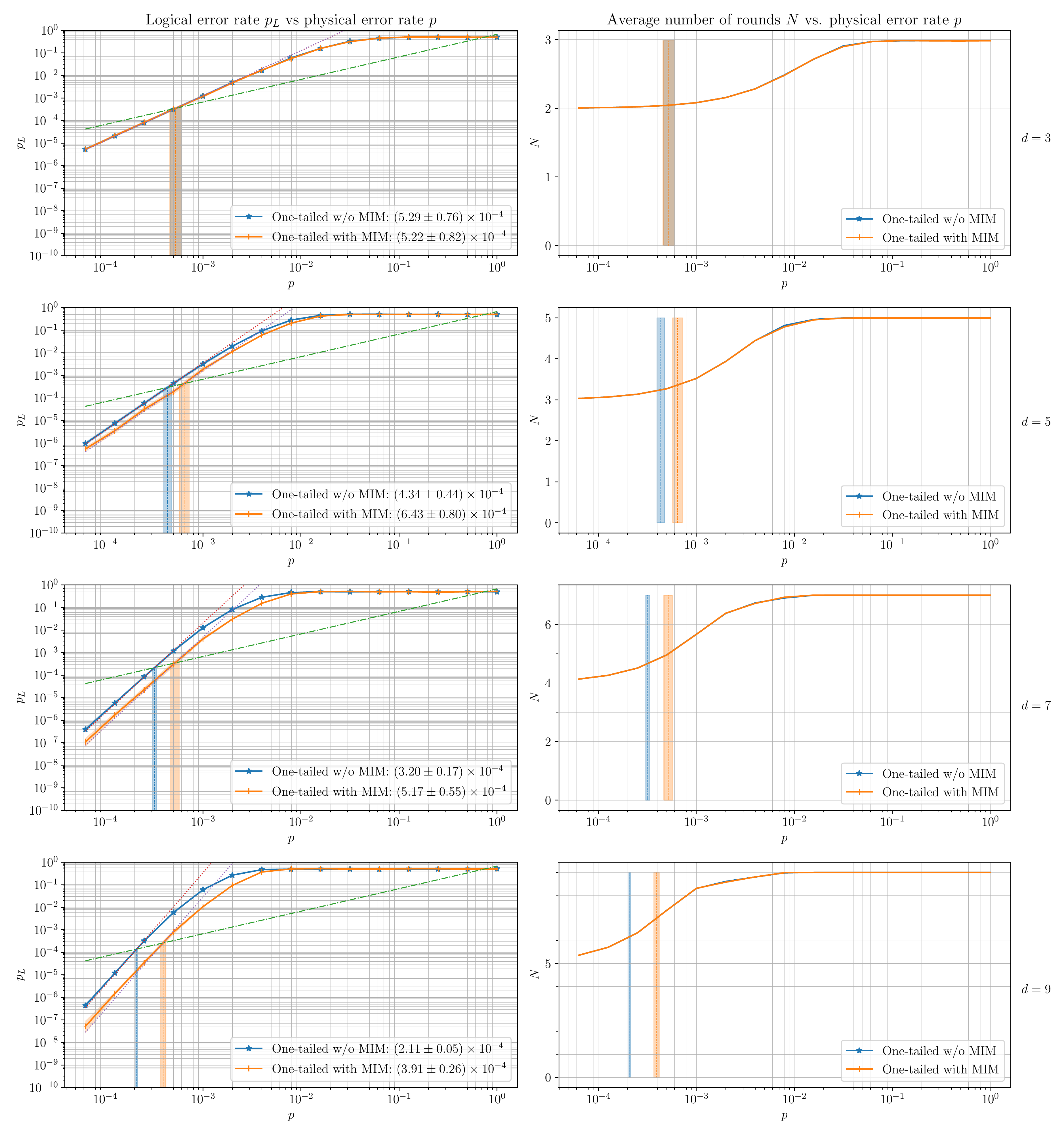}
	\caption{The effect of the MIM technique on one-tailed adaptive time decoder for hexagonal color codes of distances 3, 5, 7, and 9. The improvement is increasing with distance, with no improvement at $d=3$ and the biggest one at $d=9$.}
	\label{fig: app mim 1t rates}
\end{figure*}

\begin{figure*}[tbph]
	\centering
	\includegraphics[width=1.0\textwidth]{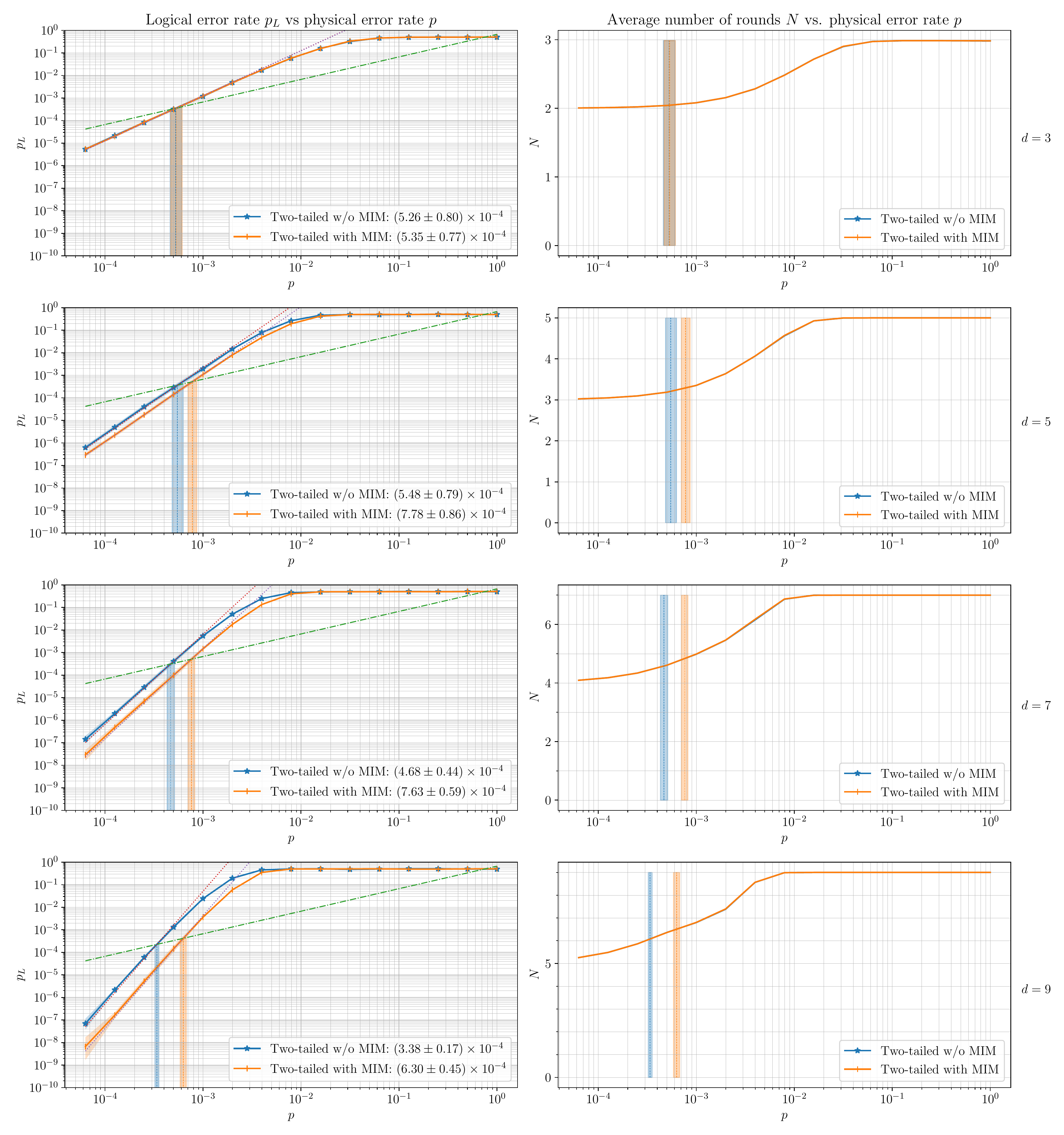}
	\caption{The effect of the MIM technique on two-tailed decoder for hexagonal color codes of distances 3, 5, 7, and 9. The improvement is increasing with distance, with no improvement at $d=3$ and the biggest one at $d=9$.}
	\label{fig: app mim 2t rates}
\end{figure*}

\begin{figure*}[tbph]
	\centering
	\includegraphics[width=1\textwidth]{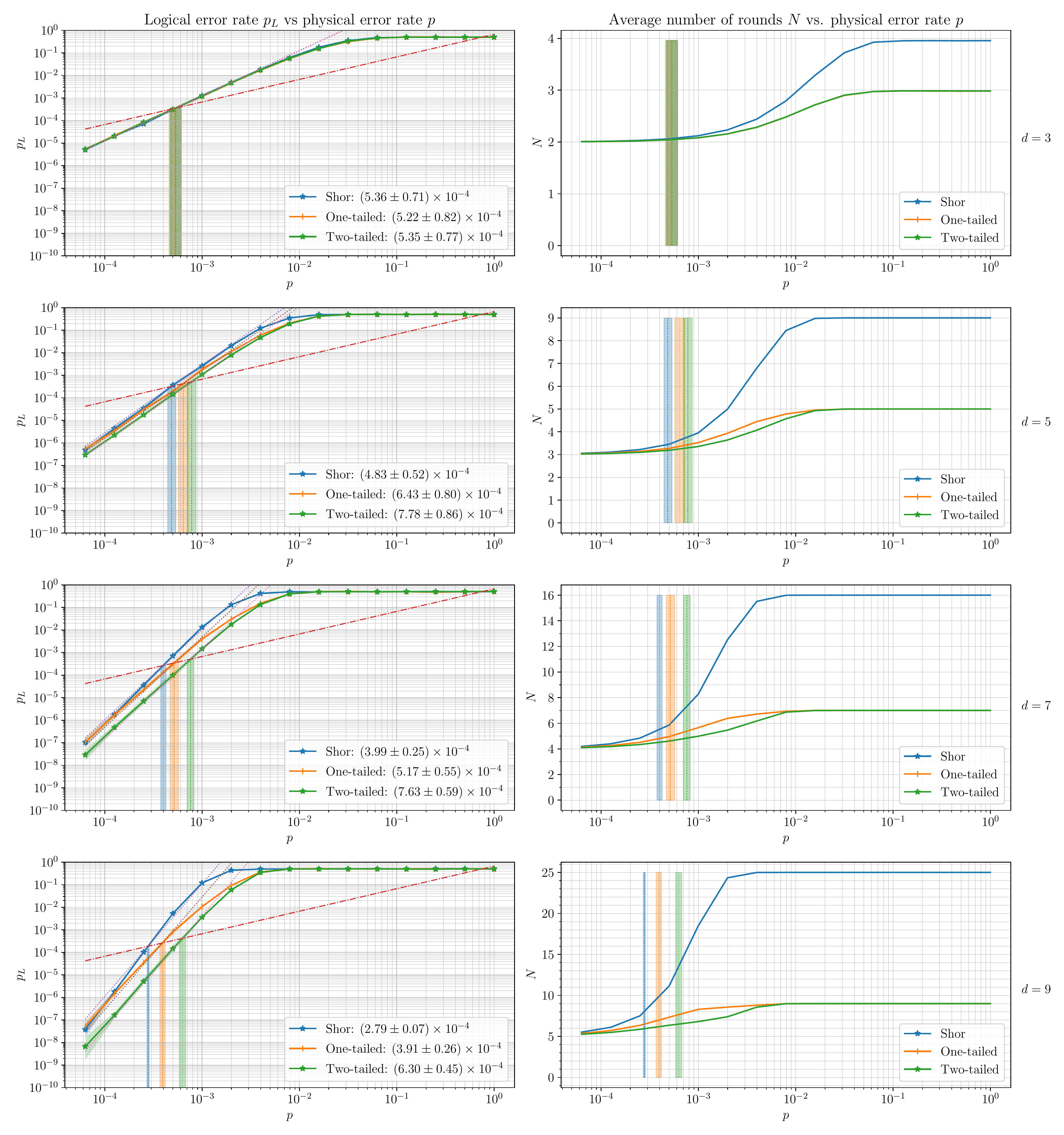}
	\caption{Comparison of one-tailed and two-tailed adaptive time decoders to Shor time decoder for hexagonal color codes of distances 3, 5, 7, and 9. The improvement is increasing with distance, with no improvement at $d=3$ and the biggest one at $d=9$. Here, all decoders use MIM-enhanced space decoding.}
	\label{fig: app sfa}
\end{figure*}

\begin{figure*}[tbph]
	\centering
	\includegraphics[width=1.0\textwidth]{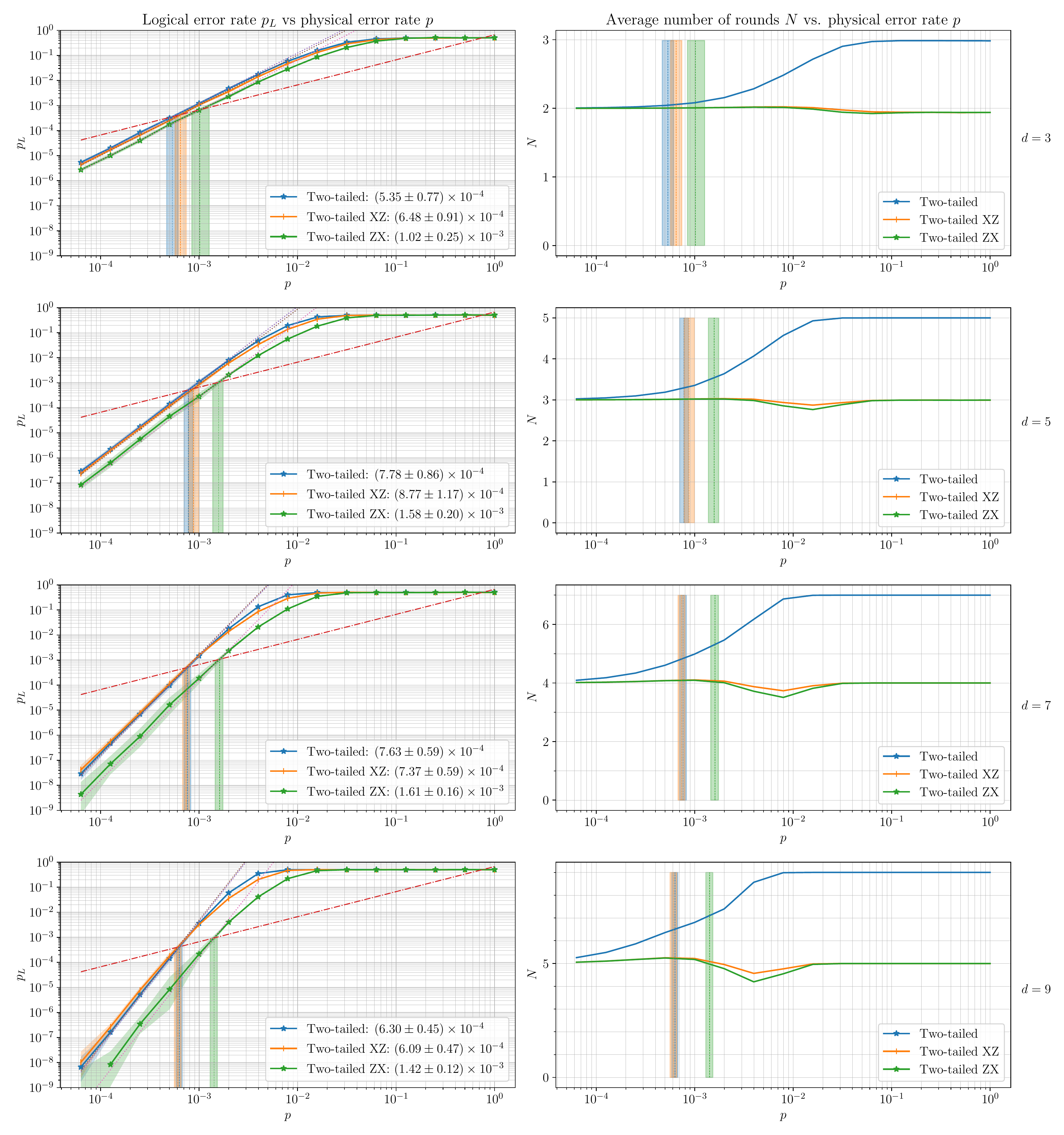}
	\caption{Comparison of the two-tailed time decoder with joint measurements and two-tailed time decoders with XZ and ZX strategies for hexagonal color codes of distances 3, 5, 7, and 9. Here, all decoders use space decoding with MIM.}
	\label{fig: app zx}
\end{figure*}

\end{document}